\documentclass[reqno, 11pt]{amsart}
\pdfoutput=1

\usepackage{amsmath,amssymb,amsthm,amsfonts}
\usepackage{xcolor}
\usepackage{mathrsfs}
\usepackage{hyperref}
\usepackage[final]{showkeys}
\usepackage[normalem]{ulem} 
\usepackage{tikz}
\usetikzlibrary{arrows}
\usetikzlibrary{math}

\newtheorem{theorem}{Theorem}[section]
\newtheorem{lemma}[theorem]{Lemma}
\newtheorem{prop}[theorem] {Proposition}
\newtheorem{cor}[theorem]  {Corollary}
\newtheorem{definition}[theorem] {Definition}

\theoremstyle{definition}

\theoremstyle{remark}
\newtheorem*{remark}{Remark}
\newtheorem*{example}{Example}

\newcommand{\e}{\mathrm{e}} 
\newcommand{\N}{\mathbb{N}}
\newcommand{\R}{\mathbb{R}}
\newcommand{\Z}{\mathbb{Z}}

\renewcommand{\P}{\mathbb{P}}
\newcommand{\E}{\mathbb{E}}
\newcommand{\dd}{\mathrm{d}} 

\newcommand{\be}{\begin{equation}}
\newcommand{\ee}{\end{equation}}

\def\1{{\mathchoice {1\mskip-4mu\mathrm l}      
{1\mskip-4mu\mathrm l}
{1\mskip-4.5mu\mathrm l} {1\mskip-5mu\mathrm l}}}

\begin{document}

\title{Hierarchical cubes: Gibbs measures and decay of correlations} 
\date{\today}

\author[S.~Jansen]{Sabine Jansen}
\address{Mathematisches Institut, Ludwig-Maximilians-Universit{\"a}t, 80333 M{\"u}nchen; Munich Center for Quantum Science and Technology (MCQST), Schellingstr.~4, 80799 M{\"u}nchen, Germany.}
\email{jansen@math.lmu.de}

\author[J.~P.~Neumann]{Jan Philipp Neumann}
\address{Mathematisches Institut, Ludwig-Maximilians-Universit{\"a}t, 80333 M{\"u}nchen, Germany.}
\email{neumann@math.lmu.de}

\begin{abstract}
	We study a hierarchical model of non-overlapping cubes of sidelengths $2^j$, $j\in \Z$. The model allows for cubes of arbitrarily small size and the activities need not be translationally invariant. It can also be recast as a spin system on a tree with a long-range hard-core interaction. We prove necessary and sufficient conditions for the existence and uniqueness of Gibbs measures, discuss fragmentation and condensation, and prove bounds on the decay of two-point correlation functions.
	
	\bigskip 
		
    
	\medskip     
    
\end{abstract}

\maketitle

\section{Introduction}  \label{sec:intro}

\noindent The present article is an addendum to the article \cite{jansen2020hierarchical}. There, the first author introduced a discrete polymer model of non-overlapping cubes of unbounded size. At every scale $j \in \N_0$, one places a tiling of $\Z^d$ consisting of translates of the prototypical $j$-block $\{1, \ldots, 2^j\}^d$ in such a way that different blocks only overlap if one is contained in the other. The model is inspired by hierarchical structures in the theory of the renormalisation group \cite{dyson1996, brydges2009}.

Here we complement the analysis from \cite{jansen2020hierarchical} by addressing infinite-volume Gibbs measures and decay of correlations. In addition, we extend the model in two directions: (a) the activities need not be scale-wise constant (same-size cubes may have different activities); (b) we allow for arbitrarily small cubes of side-lengths $2^{-j}$, $j\in \N$. The second extension is motivated by Mandelbrot's fractal percolation model, see e.g.\ \cite{mandelbrot1982, chayes^2durret1988, klattwinter2020, klattwinter2023}, and the question whether the random sets in that model can be  constructed with Gibbs measures.

Our central result, Theorem~\ref{thm:existence}, gives necessary and sufficient criteria for the existence of a Gibbs measure and says that, if a Gibbs measure exists, it is unique. The criteria are formulated as summability conditions on (effective) activities. Propositions~\ref{prop:fragmentation} and~\ref{prop:condensation} link these conditions to the physical phenomena of fragmentation and condensation in which mass escapes to microscopic cubes (fragmentation) or macroscopic cubes (condensation). Roughly, an infinite-volume Gibbs measure exists if and only if there is neither fragmentation nor condensation. In particular, at the point $\mu=\mu_c$ of first-order phase transitions in  \cite[Section~5]{jansen2020hierarchical}, the Gibbs measure is unique -- non-differentiability of the pressure and non-uniqueness of Gibbs measures are not equivalent, contrary to the Ising model, see \cite[Theorem~3.34]{friedlivelenik2017}.

An interesting subtlety shows up for small cubes and inhomogeneous activities: Theorem~\ref{thm:existence} covers situations in which some finite-volume partition functions are infinite but nevertheless a Gibbs measure exists. Clearly, ``finite volumes'', as in bounded regions of $\R^d$, generally contain an infinitude of small cubes and may have infinite partition function but we can still define a Gibbs measure with the GNZ conditions from the theory of Gibbs point processes (\cite[Theorems~2 and 2']{nguyenzessin1979}, \cite[Equation~(1)]{betsch2023gibbspointprocess}), see Definition~\ref{def:gibbs}. (An alternative point of view, based on a reinterpretation of our model as a spin system on a tree, is sketched below.) With this definition already the existence of finite-volume Gibbs measures becomes a non-trivial issue.

As a minor observation, Mandelbrot's fractal percolation is not given by one of our Gibbs measures (Proposition~\ref{prop:mandelbrot}).

Turning back to the scale-wise homogeneous setting, we investigate decay of correlations (Theorems~\ref{thm:decay-hom-general} and~\ref{thm:decay-hom-parametric}). This is motivated by standard results for spin systems on $\Z^d$, where we expect exponential decay of truncated correlations away from phase transitions and slower decay at the point of phase transition, e.g.\ algebraic at critical points (\cite{friedlivelenik2017, duminilcopin2020isingsurvey, benesetal2020decorrelation}); ``exponential'' refers to bounds of the type $C\exp(-\gamma||i-j||)$ as the Euclidean distance $||i-j||$ between the sites indexing spins $\sigma_i$, $\sigma_j$ goes to infinity. In our model the Euclidean distance is replaced by an ultrametric distance $D(B,B')= 2^{d\,\mathrm{lcs}(B,B')}$ between blocks, with $\mathrm{lcs}(B,B')\in \Z$ the scale of the smallest block containing both $B$ and $B'$. In the gas phase, correlations decay as $\exp(-\mathrm{const} \, D(B,B'))$ (Theorem~\ref{thm:decay-hom-general}); for a concrete class of models with a first-order phase transition, the decay at the point of phase transition is instead subexponentially bounded by $\exp(- \mathrm{const} \, D(B,B')^\alpha)$, $\alpha \in (0, 1)$, (Theorem~\ref{thm:decay-hom-parametric}).

To conclude we highlight an alternative point of view on our model as a spin system on a tree. The set of admissible cubes has a natural tree structure. By drawing edges from any block to those it contains one scale below, our cube set becomes a regular $2^d$-ary tree (also called Cayley tree or Bethe lattice of degree/branching number $2^d$) whose edges we prefer to think of as oriented toward decreasing scales, see Figure~\ref{fig:1d} (cf.\ also \cite{kendallwilson2003quadtrees} with regard to this tree structure arising via the successively refined tiling of $\R^d$ with cubes). Each vertex of the tree comes with a $\{0,1\}$-valued spin variable. The interaction is a long-range hard-core interaction: if a vertex (block) is occupied (spin~$1$), then all descendants and ancestors in the tree must be unoccupied (spin~$0$).

There is an extensive body of literature related to the theory of Gibbs measures on trees and their relationship with tree-indexed Markov chains, the fundamentals of which can be found in Georgii \cite[Chapter~12]{georgii2011gibbslatticespin}. More recent developments include the articles \cite{coquillekuelskeleny2023, coquillekuelskeleny2023+} by Coquille, K\"ulske and Le Ny. The general theory is not  applicable to our context because of the interaction's long range. Our Gibbs measures are not described as tree-indexed Markov chains, instead we provide a concrete characterisation as \emph{hierarchical measures} (defined in Lemma~\ref{lem:hierarchical}).

Our results easily generalise to similar trees. The most obvious  modification simply replaces the cube sidelength base/subdivision parameter $M = 2$ by any larger integer. 
All our results then carry over mutatis mutandis and, in particular, the relationship of scaling limits and fractal/Mandelbrot percolation with its Galton-Watson construction remains intact.

The remainder of this article is organised as follows. In Section~\ref{sec:results} we first give the formal definition of our model and then formulate our main result, Theorem~\ref{thm:existence}, and the observation on Mandelbrot percolation not being representable via one of our Gibbs measures (Subsection~\ref{subsec:existence}). Next we explain the relation with condensation and fragmentation in Subsection~\ref{subsec:cond-frag} and analyse decay of correlations in Subsection~\ref{subsec:decay}. Proofs are given in Section~\ref{sec:proofs}.

\section{Main results} \label{sec:results}

\subsection{Existence and uniqueness of infinite-volume Gibbs measures} \label{subsec:existence}

Fix a dimension $d\in \N$. We define the set of admissible blocks to be the union $\mathbb B = \bigcup_{j \in \Z} \mathbb B_j$ 
where 
\[
	\mathbb B_j : = \bigl\{ \mathbf k + [0,2^j)^d \mid \mathbf k \in 2^j \N_0^d \bigr\}
\] 
denotes the set of blocks of scale $j$ or simply $j$-blocks. Thus $\mathbb B_j$ consists of non-overlapping cubes of sidelength $2^j$ all contained in the non-negative orthant $\R_+^d = [0,\infty)^d$. Allowing for $\mathbf k \in 2^j \Z^d$ instead results in a model on $\R^d$ in which different orthants do not interact. Cubes may be arbitrarily small or large -- we do not impose any bounds on $j$. 

Blocks interact through a hard-core potential that forbids overlap. Thus $B \in \mathbb B$ cannot occur alongside any additional block $B'$ from
\[
	\mathbb I_B := \{ B' \in \mathbb B \mid B' \cap B \neq \varnothing \},
\] 
the set of blocks that intersect $B$. Prohibiting, in particular, block multiplicities other than $0$ or $1$, our configuration space $\Omega = \mathcal P(\mathbb B)$ consists of the collection of all subsets $\mathcal B \subset \mathbb B$. It is equipped with the $\sigma$-algebra generated by the sets $\{\mathcal B \in \Omega \mid \mathcal B \ni B\}$, $B\in \mathbb B$. 

\begin{figure}
\centering
\begin{tikzpicture}
	\draw [help lines] (0,1) -- ++(0,-8) -- ++(10,0) -- ++(0,8) -- cycle;
	\draw[->] (-.5,-6.5) -- (10.5,-6.5); 
	\begin{scope}[color=gray!50] 
		\draw[dash dot] (10,1) -- (5,0); 
		\foreach \j in {0,...,5} 
			{
				\foreach \k [parse=true] in {0,...,2^\j-1} 
					{
						\draw[[-), line width=1 pt] (10*\k*.5^\j,-\j) -- +(10*.5^\j,-0); 
						\fill (10*\k*.5^\j+5*.5^\j,-\j) circle (2 pt); 
						\draw[dash dot] (10*\k*.5^\j+5*.5^\j,-\j) -- +(-2.5*.5^\j,-1); 
						\draw[dash dot] (10*\k*.5^\j+5*.5^\j,-\j) -- +(2.5*.5^\j,-1); 
					}
			}
		\end{scope}
		\begin{scope}[color=green] 
			\tikzmath{\j=2;} 
			\foreach \k in {1} 
				{
					\draw[[-), line width=1 pt] (10*\k*.5^\j,-\j) -- +(10*.5^\j,-0); 
					\fill (10*\k*.5^\j+5*.5^\j,-\j) circle (2 pt); 
					\fill[opacity=.1] (10*\k*.5^\j,1) -- ++(0,-8) -- ++(10*.5^\j,0) -- ++(0,8) -- cycle; 
					\draw[[-), line width=1 pt] (10*\k*.5^\j,-6.5) -- +(10*.5^\j,-0); 
				}
		\end{scope}
		\begin{scope}[color=cyan]
			\tikzmath{\j=3;}
			\foreach \k in {0,4,6}
				{
					\draw[[-), line width=1 pt] (10*\k*.5^\j,-\j) -- +(10*.5^\j,-0);
					\fill (10*\k*.5^\j+5*.5^\j,-\j) circle (2 pt);
					\fill[opacity=.1] (10*\k*.5^\j,1) -- ++(0,-8) -- ++(10*.5^\j,0) -- ++(0,8) -- cycle;
					\draw[[-), line width=1 pt] (10*\k*.5^\j,-6.5) -- +(10*.5^\j,-0);
			}
		\end{scope}
		\begin{scope}[color=blue]
			\tikzmath{\j=4;}
			\foreach \k in {2,11}
				{
					\draw[[-), line width=1 pt] (10*\k*.5^\j,-\j) -- +(10*.5^\j,-0);
					\fill (10*\k*.5^\j+5*.5^\j,-\j) circle (2 pt);
					\fill[opacity=.1] (10*\k*.5^\j,1) -- ++(0,-8) -- ++(10*.5^\j,0) -- ++(0,8) -- cycle;
					\draw[[-), line width=1 pt] (10*\k*.5^\j,-6.5) -- +(10*.5^\j,-0);
			}
		\end{scope}
		\begin{scope}[color=violet]
			\tikzmath{\j=5;}
			\foreach \k in {7,21,30,31}
				{
					\draw[[-), line width=1 pt] (10*\k*.5^\j,-\j) -- +(10*.5^\j,-0);
					\fill (10*\k*.5^\j+5*.5^\j,-\j) circle (2 pt);
					\fill[opacity=.1] (10*\k*.5^\j,1) -- ++(0,-8) -- ++(10*.5^\j,0) -- ++(0,8) -- cycle;
					\draw[[-), line width=1 pt] (10*\k*.5^\j,-6.5) -- +(10*.5^\j,-0);
			}
		\end{scope}
\end{tikzpicture}
\caption{A hierarchical cube configuration in dimension $d=1$. The top part shows a truncated portion of the set $\mathbb B$ of admissible cubes (intervals) with mid-point nodes and dashed lines emphasising the implicit binary tree structure. Occupied cubes are highlighted in scale-dependent colours, both in the tree structure (top) and on the real line (bottom). The shaded areas visualise excluded volume effects.} \label{fig:1d}
\end{figure}

\begin{figure}
\centering
\begin{tikzpicture}
	\draw [help lines] (0,0) -- ++(10,0) -- ++(0,10) -- ++(-10,0) -- cycle;
	\draw[->] (-.5,0) -- (10.5,0); 
	\draw[->] (0,-.5) -- (0,10.5); 
	\begin{scope}[color=yellow] 
		\tikzmath{\j=1;} 
		\foreach \k [parse=true] in {(0*10*.5^\j,1*10*.5^\j)} 
			{
				\fill \k circle (2pt); 
				\draw[line width=2 pt] \k -- +(10*.5^\j,0); 
				\draw[line width=2 pt] \k -- +(0,10*.5^\j); 
				\fill[opacity=.5] \k -- ++(10*.5^\j,0) -- ++(0,10*.5^\j) -- ++(-10*.5^\j,0) -- cycle; 
			}
	\end{scope}
	
	\begin{scope}[color=green]
		\tikzmath{\j=2;}
		\foreach \k [parse=true] in {(0*10*.5^\j,0*10*.5^\j), (10*.5^\j,0*10*.5^\j), (3*10*.5^\j,1*10*.5^\j)}
			{
				\fill \k circle (2pt);
				\draw[line width=2 pt] \k -- +(10*.5^\j,0);
				\draw[line width=2 pt] \k -- +(0,10*.5^\j);
				\fill[opacity=.5] \k -- ++(10*.5^\j,0) -- ++(0,10*.5^\j) -- ++(-10*.5^\j,0) -- cycle;
			}
	\end{scope}
	
	\begin{scope}[color=cyan]
		\tikzmath{\j=3;}
		\foreach \k [parse=true] in {(5*10*.5^\j,1*10*.5^\j), (4*10*.5^\j,4*10*.5^\j), (6*10*.5^\j,4*10*.5^\j), (4*10*.5^\j,6*10*.5^\j), (7*10*.5^\j,6*10*.5^\j), (5*10*.5^\j,7*10*.5^\j)}
			{
				\fill \k circle (2pt);
				\draw[line width=2 pt] \k -- +(10*.5^\j,0);
				\draw[line width=2 pt] \k -- +(0,10*.5^\j);
				\fill[opacity=.5] \k -- ++(10*.5^\j,0) -- ++(0,10*.5^\j) -- ++(-10*.5^\j,0) -- cycle;
			}
	\end{scope}
	
	\begin{scope}[color=blue]
		\tikzmath{\j=4;}
		\foreach \k [parse=true] in {(13*10*.5^\j,1*10*.5^\j), (14*10*.5^\j,1*10*.5^\j), (3*10*.5^\j,4*10*.5^\j), (2*10*.5^\j,5*10*.5^\j), (5*10*.5^\j,5*10*.5^\j), (8*10*.5^\j,5*10*.5^\j), (3*10*.5^\j,6*10*.5^\j), (6*10*.5^\j,6*10*.5^\j), (15*10*.5^\j,10*10*.5^\j), (12*10*.5^\j,11*10*.5^\j), (11*10*.5^\j,12*10*.5^\j), (14*10*.5^\j,14*10*.5^\j), (13*10*.5^\j,15*10*.5^\j)}
			{
				\fill \k circle (2pt);
				\draw[line width=2 pt] \k -- +(10*.5^\j,0);
				\draw[line width=2 pt] \k -- +(0,10*.5^\j);
				\fill[opacity=.5] \k -- ++(10*.5^\j,0) -- ++(0,10*.5^\j) -- ++(-10*.5^\j,0) -- cycle;
			}
	\end{scope}
\end{tikzpicture}
\caption{A hierarchical cube configuration in dimension $d=2$ consisting of occupied cubes (squares) on finitely many colour-coded scales. The small points in the squares' bottom left corners mark the respective $\mathbf k$-parameters from the definition of $\mathbb B$, the squares' closed sides (bottom and left) are emphasised with darker lines.} \label{fig:2d}
\end{figure}

In the following, the symbol $\omega$ stands not only for elements $\omega \in \Omega$ but also for the configuration-valued random variable given by the identity map $\omega \mapsto \omega$. This allows us to use notational shorthand from probability theory, e.g., the above generating events are written as $\{\omega \ni B\}$. 

We define infinite-volume Gibbs measures through the GNZ equation (\cite[Theorem~2, Eq.~(3.3) or Theorem~2, Cond.~(3)]{nguyenzessin1979}, \cite[Eq.~(1)]{betsch2023gibbspointprocess}), borrowed from the theory of Gibbs point processes rather than the usual DLR conditions. We discuss the relation between the definitions toward the end of this subsection. 

\begin{definition} \label{def:gibbs}
	A probability measure $\mathbb P$ on $\Omega$ is a Gibbs measure for the activity $z: \mathbb B \to \R_+$ if 
	\[
		\E \Bigl[ \sum_{B \in \omega} F(B, \, \omega) \Bigr] = \sum_{B \in \mathbb B} z(B) \, \E \Bigl[ \1_{ \{ \omega \cap \mathbb I_B = \varnothing \} } \, F(B, \, \omega \cup \{B\}) \Bigr]. 
\] 
	for all measurable $F:\mathbb B\times \Omega\to \R_+$. The set of Gibbs measures is denoted $\mathcal G(z)$. 
\end{definition} 

\noindent By standard arguments, it is enough to consider functions $F$ that are indicators of measurable sets of the form $\{B_0\} \times \mathcal A \subset \mathbb B \times \Omega$. Moreover, $F$ is only ever evaluated on pairs $(B, \, \mathcal B)$ with $B \in \mathcal B$ so $\mathcal A$ can be chosen to only depend on the configuration outside $\{B_0\}$. Thus, $\P$ is in $\mathcal G(z)$ if and only if 
\begin{equation} \label{eq:gnz2}
		\P \bigl( \omega \ni B, \, \omega \setminus \{B\} \in \mathcal A \bigr) = z(B) \, \P \bigl( \omega \cap \mathbb I_B = \varnothing, \, \omega \setminus \{B\} \in \mathcal A \bigr)
\end{equation} 
for all $B \in \mathbb B$ and all measurable subsets $\mathcal A \subset \Omega$. Of course, every Gibbs measure is supported on the non-overlap event 
\[
	\Delta : = \{ \mathcal B \in \Omega \mid \forall B,B' \in \mathcal B: \, B \neq B' \Rightarrow  B \cap B' = \varnothing \}.
\]
The grand-canonical \emph{partition function} of a cube $\Lambda \in \mathbb B$ is a sum over configurations boxed in by $\Lambda$, i.e., elements of the space
\[
	\Omega_\Lambda : = \{\mathcal B \in \Omega \mid \forall B \in \mathcal B: B \subset \Lambda\}.
\] 
Accounting for the activity and the non-overlap prescription, it is given by 
\[
	\Xi_\Lambda(z) : = \sum_{\mathcal B \in \Omega_\Lambda} \1_\Delta(\mathcal B) \prod_{B \in \mathcal B} z(B). 
\]
The empty product is $1$ by convention. In view of 
\begin{equation}\label{eq:xifi-zfi}
	1 + \sum_{B \subset \Lambda} z(B) \leq \Xi_\Lambda(z)  \leq \prod_{B \subset \Lambda}\bigl( 1 + z(B) \bigr),
\end{equation} 
the partition function of $\Lambda$ is finite if and only if $\sum_{B \subset \Lambda} z(B) < \infty$. This also allows restricting the sum defining $\Xi_\Lambda(z)$ to the (countably many) finite configurations in $\Omega_\Lambda$, a fact we use without further comment.

Given a fixed block $B \in \mathbb B$ of scale $j \in \Z$, its \emph{effective activity}, cf.\ \cite[Section~3]{jansen2020hierarchical}, is 
\begin{equation} \label{eq:zhat}
	\widehat z(B) : = \frac{z(B)}{\prod_{B' \in \mathbb B_{j-1}: B' \subset B} \Xi_{B'}(z)}.
\end{equation}
Here and throughout the rest of the paper, we employ the convention that fractions with infinite denominators are zero.

If, on every individual scale $j \in \Z$, the activity assigns the same value to all $j$-blocks, we call the activity \emph{homogeneous} or \emph{scale-wise constant} and write $z_j$ and $\widehat z_j$ rather than $z(B)$ and $\widehat z(B)$ for $B \in \mathbb B_j$, $j \in \Z$. 

\begin{theorem} \label{thm:existence-hom}
	If the activity is homogeneous, then $\mathcal G(z) \neq \varnothing$ if and only if 
	\[
		\sum_{j \in \N_0} 2^{dj}z_{-j} < \infty \quad \text{ and } \quad \sum_{j \in \N_0} \widehat z_j < \infty.
	\] 	
	If these conditions are satisfied, the Gibbs measure is unique: $|\mathcal G(z)| = 1$. 
\end{theorem} 

\noindent The second summability condition in Theorem~\ref{thm:existence-hom} and its relation to the phenomenon of condensation was discussed at length in \cite{jansen2020hierarchical}. There, negative scales were excluded, which amounts to $z_j=0$ for all $j<0$. The first summability condition,  which simply ensures that every block has finite partition function, is trivial in that situation.

Noting that ``anchoring'' the summability conditions of Theorem~\ref{thm:existence-hom} at scales other than $0$ changes nothing, said theorem is an immediate consequence of the following general statement. Indeed, the first summability condition above, implying $\Xi_\Lambda(z) < \infty$ for all $\Lambda \in \mathbb B$, makes Condition~(i) below vacuously true while the second summability condition above translates effortlessly into the subsequent Condition~(ii).

\begin{theorem} \label{thm:existence}
	$\mathcal G(z) \neq \varnothing$ if and only if the following two conditions hold true: 
	\begin{itemize}
		\item [(i)] For all $\Lambda \in \mathbb B$ with $\Xi_\Lambda(z) = \infty$, 
			\[
				\sum_{\substack{ B\subset \Lambda: \\ \Xi_B(z) < \infty }} z(B) = \infty. 
			\] 
		\item [(ii)] For every $B\in \mathbb B$, 
			\[
				\sum_{B' \supset B} \widehat z(B') < \infty. 
			\]
	\end{itemize} 
	If these conditions are satisfied, the Gibbs measure is unique: $|\mathcal G(z)| = 1$. 
\end{theorem}

\noindent Condition~(i) is obviously satisfied when all finite-volume partition functions $\Xi_\Lambda(z)$, $\Lambda \in \mathbb B$, are finite. Perhaps surprisingly, however, Condition~(i) does allow for $\Xi_\Lambda(z) = \infty$, as long as the subcubes of $\Lambda$ with finite partition function carry enough mass. In the proof of Lemma~\ref{lem:infinite-pf}, we show that this happens precisely when $\Lambda$ has disjoint subcubes $\Lambda_1, \Lambda_2, \ldots$ with finite partition functions $\Xi_{\Lambda_j}(z) < \infty$ for all $j \in \N$ but infinite product $\prod_{j \in \N} \Xi_{\Lambda_j}(z) = \infty$. 

Regarding Condition~(ii), note that its summability assertion only needs to be checked for a single $B \in \mathbb B$ since the set $\{B' \in \mathbb B \mid B' \supset B\}$ changes only by finitely many blocks when switching between any two choices for $B$. This is because we work on $\R_+^d$. On $\R^d$, we would need to check the condition for $2^d$ blocks, one for each orthant. In particular, Condition~(ii) is automatically satisfied whenever some block $\Lambda \in \mathbb B$ has infinite partition function as one then has $\widehat z(B) = 0$ for all blocks $B \supset \Lambda$. Thus, no activity can fail both Conditions~(i) and (ii) in Theorem~\ref{thm:existence}.

\begin{example}
	Let $d=1$. Partition the unit interval as 
	\[
		[0,1) = \Bigl[0,\frac12\Bigr) \cup \Bigl[\frac12, \frac12+\frac14\Bigr) \cup \cdots = \bigcup_{j\in \N_0} [s_j,s_{j+1})
	\] 
	with $s_j = \sum_{k=1}^j 2^{-k}$ for all $j \in \N_0$. Fixing some sequence $(\lambda_j)_{j \in \N_0}$ in $\R_+$, let $\mathbb P$ be the measure on $\Omega$ under which the events $\{\omega \ni B\}$, $B \in \mathbb B$, are independent with 
	\[
		\P \bigl( \omega \ni [s_j,s_{j+1}) \bigr) = \frac {\lambda_j}{1 + \lambda_j} 
	\] 
	for all $j\in \N_0$, and $\P(\omega \ni B) = 0$ for all other blocks. Then $\mathbb P$ is a Gibbs measure for the activity
	\[
		z(B) = \begin{cases}
			\lambda_j & \text{if } B=[s_j, s_{j+1}), \, j \in \N_0, \\
			0 & \text{else}.
		\end{cases}
	\]
	In particular, $\mathcal G(z) \neq \varnothing$ even if $\Xi_{[0,1)}(z) = \prod_{j \in \N_0} (1 + \lambda_j) = \infty$. In the latter case the activity $z$ is not uniquely determined by the measure $\P$: Changing $z$ in an arbitrary way on the cubes $B$ from $	\{[s_j, 1) \mid j \in \N_0\} \cup \{[0, 2^j) \mid j \in \N\}$ -- i.e., those with infinite partition function $\Xi_B(z)=\infty$ -- results in a new activity $z'$ for which $\P\in \mathcal G(z')$ as well. Of course, similar constructions are possible in any dimension $d \in \N$.
\end{example}

\begin{remark}
	When the infinite-volume Gibbs measure exists, volumes with infinite partition function may be identified via the chain of equivalences 
	\begin{align*}
		\Xi_\Lambda(z) = \infty & \Leftrightarrow \P \bigl( \{ B \in \omega \mid B \subset \Lambda \} = \varnothing \bigr) = 0 \\
		& \Leftrightarrow \P \bigl( |\{ B \in \omega \mid B \subset \Lambda \}| < \infty \bigr) < 1 \\
		& \Leftrightarrow \P \bigl( | \{ B \in \omega \mid B \subset \Lambda \} | = \infty \bigr) = 1
	\end{align*}
	that we mention without proof. The first equivalence is interesting because $1/\Xi_\Lambda(z)$ is often thought of as an emptiness probability, thus infinite partition functions correspond to zero probability for being empty. The last equivalence says that in fact such blocks then automatically contain infinitely many blocks. 
\end{remark} 

\noindent Let us briefly comment on our definition of Gibbs measures. Definition~\ref{def:gibbs} corresponds to the GNZ equation for Gibbs point processes. As noted in the introduction, they can be formulated even when some finite-volume partition function $\Xi_\Lambda(z)$ is infinite, in which case a naive use of the usual DLR conditions would be problematic.

However, the DLR conditions do make sense if we change perspectives, identify a configuration $\mathcal B\in \Omega$ with its indicator function $\1_\mathcal B \in \{0,1\}^\mathbb B$, and view our model as a spin system on a lattice. ``Finite volume'' then no longer refers to bounded regions in $\R_+^d$ but instead to finite subsets $\mathcal L\subset \mathbb B$. The DLR condition, see e.g.\ \cite[Definitions~(1.23) and (2.9)]{georgii2011gibbslatticespin}, associated to a singleton reference volume $\mathcal L = \{B\}$ boils down to the following requirement: 
\begin{align}
	& \P \bigl( \omega \ni B, \, \omega \setminus \{B\} \in \mathcal A \bigr) \notag \\
	& \qquad = \frac{z(B)}{1 + z(B)} \, \P \bigl( \omega \cap \mathbb I_B \setminus \{B\} = \varnothing, \, \omega\setminus \{B\} \in \mathcal A \bigr) \label{eq:dlr-singleton}
\end{align}
for all $B\in \mathbb B$ and all measurable $\mathcal A \subset \Omega$, where we note that the denominator $Z_{\{B\}} = 1 + z(B)$ is just the (finite) partition function in finite (lattice) volume $\mathcal L = \{B\}$ (with empty boundary condition). This is easily seen to be equivalent to Eq.~\eqref{eq:gnz2}: Given $\P$, $z$, $B$, $\mathcal A$, and substituting the latter by
\[
	\{\omega \in \Omega \mid \omega \cap \mathbb I_B = \varnothing\} \cap \mathcal A,
\]
Eq.~\eqref{eq:gnz2} implies
\begin{align*}
	& \P \bigl( \omega \cap \mathbb I_B \setminus \{B\} = \varnothing, \, \omega \setminus \{B\} \in \mathcal A \bigr) \\
	& \qquad = (1 + z(B)) \, \P \bigl( \omega \cap \mathbb I_B = \varnothing, \, \omega \setminus \{B\} \in \mathcal A \bigr)
\end{align*}
while, conversely, Eq.~\eqref{eq:dlr-singleton} yields
\begin{align*}
	& \P \bigl( \omega \cap \mathbb I_B = \varnothing, \, \omega \setminus \{B\} \in \mathcal A \bigr) \\
	& \qquad = \left( 1 - \frac{z(B)}{1 + z(B)} \right) \, \P \bigl( \omega \cap \mathbb I_B \setminus \{B\} = \varnothing, \, \omega \setminus \{B\} \in \mathcal A \bigr) \\
	& \qquad = \frac{1}{1 + z(B)} \, \P \bigl( \omega \cap \mathbb I_B \setminus \{B\} = \varnothing,\, \omega \setminus \{B\} \in \mathcal A \bigr).
\end{align*}
Thus, Definition~\ref{def:gibbs} is equivalent to the single-site DLR conditions for the model treated as a lattice spin system. It is known in some cases that the single-site DLR conditions already imply the full DLR conditions, see \cite[Theorem~1.33]{georgii2011gibbslatticespin} (that general theorem is not applicable here, nevertheless our measures do satisfy the full set of DLR equations). Interestingly Georgii presented the GNZ characterisation as a continuum analogue to that fact, see the first sentence of \cite[Section~3]{georgii1976canonical}.

We conclude this section with an example that is not a Gibbs measure. The example is related to Mandelbrot's percolation process, cf.\ \cite{mandelbrot1982, chayes^2durret1988, klattwinter2020, klattwinter2023}, and it was proposed as a candidate scaling limit for finite-volume Gibbs measures in the context of continuous phase transitions in \cite[Section 5.3]{jansen2020hierarchical}. Let $p\in (0,1]$. Define a measure on configurations of blocks contained in $[0,1)^d$ as follows: Pick the top-level block $[0,1)^d$ with probability $p$. If the block has been chosen, stop. Otherwise, for each of the $2^d$ blocks $B\in \mathbb B_{-1}$ contained in $[0,1)^d$, decide independently for each of them whether to keep them or not, and then move to the next scale below. After $n$ iterations, we obtain a random collection $\mathcal B_n \subset \mathbb B_0\cup \cdots \cup \mathbb B_{-n}$. We let $\mathcal B:= \bigcup_{n\in \N_0} \mathcal B_n$ and let $\mathbb P_p$ be the distribution of $\mathcal B$.

\begin{prop} \label{prop:mandelbrot}
	There is no activity $z: \mathbb B \to \R_+$ such that $\P_p \in \mathcal G(z)$. 
\end{prop}

\begin{proof}
	The GNZ Eq.~\eqref{eq:gnz2} with $\mathcal A = \Omega$ reads
	\[
		\P(\omega \ni B) = z(B) \, \P(\omega \cap \mathbb I_B = \varnothing)
	\]
	which, upon choosing $\P = \P_p$ and $B = [0, 1)^d$, becomes
	\[
		p = z \bigl( [0, 1)^d \bigr) \, \prod_{\substack{B' \in \mathbb B: \\ B' \subset [0, 1)^d}} (1 - p).
	\]
	Since $p \in (0,1]$, the former identity absurdly states $p = 0$, no matter the choice of $z$.
\end{proof}

\noindent The case $p = 0$ is excluded because $\P_0(\omega = \varnothing) = 1$ and, hence, $\P_0$ is trivially in $\mathcal G(0)$.

\subsection{Fragmentation and condensation} \label{subsec:cond-frag}

The standard procedure for constructing Gibbs measures is to consider finite-volume Gibbs measures (or specification kernels) and pass to the limit. 
In this spirit, we  may start from a finite volume $\Lambda \in \mathbb B$ and consider an upward truncation of the activity
\[
	z_\Lambda(B) : = \begin{cases} 
			z(B) & \text{if } B \subset \Lambda, \\
			0 & \text{else},
		\end{cases} 
\]
as well as downward truncations $z^{(n)}$, $n\in \Z$, given by 
\[
	z^{(n)}(B) : = \begin{cases} 
			z(B) & \text{if } B \in \bigcup_{j \geq -n} \mathbb B_j,\\
			0 & \text{else}.
		\end{cases} 
\]
By standard arguments, the (essentially finite-volume) Gibbs measure for the two-sidedly truncated activity $z_\Lambda^{(n)}$ is 
\[
	\P_{\Lambda}^{(n)} \bigl( \omega = \mathcal B \bigr) = \1_\Delta(\mathcal B) \, \frac{\prod_{B \in \mathcal B} z^{(n)}(B)}{\Xi_\Lambda(z ^{(n)})}, \quad \mathcal B \in \Omega_\Lambda.
\]
Our first result in this section addresses the limit $n \to \infty$ at fixed $\Lambda$ when Condition~(i) in Theorem~\ref{thm:existence} fails.

\begin{prop} \label{prop:fragmentation} 
	If the activity $z$ violates Condition~(i) in Theorem~\ref{thm:existence}, then there exists an infinite set $\mathbb V(z) \subset \mathbb B$ such that, for all $\Lambda, B \in \mathbb V(z)$ with $B \subset \Lambda$,
	\[
		\lim_{n\to \infty} \P_\Lambda^{(n)}(\exists B' \subset B: \omega \ni B') = \lim_{n\to \infty} \P_\Lambda^{(n)}(\exists B' \subset B: \omega \cap \mathbb V(z) \ni B') = 1
	\]
	but also
	\[
		\lim_{n\to \infty} \P_\Lambda^{(n)}(\omega \ni B) = 0, \quad \lim_{n \to \infty} \P_\Lambda^{(n)} ( \exists B' \subset B: \omega \setminus \mathbb V(z) \ni B' ) < 1.
	\]
\end{prop}

\noindent Roughly, if Condition~(i) in Theorem~\ref{thm:existence} fails, then  a positive fraction of mass is lost to ``infinitely small blocks'' along $\mathbb V(z)$ whose individual contributions all vanish in the limit. The most extreme instances of this phenomenon occur when the latter set contains all blocks $B \subset \Lambda$. In fact, this case always applies to homogeneous activities violating the first summability condition in Theorem~\ref{thm:existence-hom}.

\begin{cor} \label{cor:fragmentation-hom}
	Suppose the activity $z$ is scale-wise constant, given by $(z_j)_{j \in \Z}$, such that $\sum_{j \in \N_0} 2^{dj} z_{- j} = \infty$. Then
	\[
		\lim_{n \to \infty} \P_\Lambda^{(n)}(\omega \ni B) = 0, \quad \lim_{n \to \infty} \P_\Lambda^{(n)}(\exists B' \subset B: \omega \ni B') = 1
	\]
	for every $\Lambda \in \mathbb B$ and all blocks $B \subset \Lambda$.
\end{cor}

\noindent Our second result in this section is about the limit $\Lambda\uparrow \R_+^d$ when Condition~(ii) in Theorem~\ref{thm:existence} fails, assuming Condition~(i) holds true. It follows from Theorem~\ref{thm:existence} that, for every $\Lambda\in \mathbb B$, there is a uniquely defined Gibbs measure $\P_\Lambda\in \mathcal G(z_\Lambda)$ -- no  need to assume $\Xi_\Lambda(z)<\infty$ or to downward truncate the activity.

In the following, the limit $\Lambda\uparrow \R_+^d$ refers to limits along growing sequences $\Lambda_1 \subset \Lambda_2 \subset \ldots $  in $\mathbb B$ with $\bigcup_{n\in \N}\Lambda_n = \R_+^d$, and the limit statements hold true for every such sequence $(\Lambda_n)$.

\begin{prop} \label{prop:condensation}
	If the activity satisfies Condition~(i) but not Condition~(ii) in Theorem~\ref{thm:existence}, then, for every block $B\in \mathbb B$, 
	\[
		\lim_{\Lambda \uparrow \R_+^d} \P_\Lambda(\omega \ni B) = 0, \quad \lim_{\Lambda \uparrow \R_+^d} \P_\Lambda(\exists B' \supset B: \omega \ni B') = 1.
	\]
\end{prop} 

\noindent The interpretation is that all mass escapes to ever larger cubes and no mass at all remains with finite-size cubes. Moreover, as our configurations do not allow for infinite blocks, we see that $\P_\Lambda$ at best converges to the measure concentrated on the empty configuration. This measure is not in $\mathcal G(z)$ because the failure of Condition~(ii) from Theorem \ref{thm:existence} excludes the trivial case $z = 0$. 

\subsection{Decay of correlations} \label{subsec:decay}

In this subsection, we take a look at how occurrence events for pairs of distinct blocks correlate under a Gibbs measure when the blocks are, in a certain sense, far from each other.

For the sake of concreteness, we  consider the homogeneous case only and fix an activity $z$ via a sequence $(z_j)_{j\in \Z}$. By Theorem~\ref{thm:existence-hom}, the existence of a (necessarily unique) Gibbs measure for $z$ also implies that both the \emph{pressure}
\[
	p : = \lim_{j \to \infty} 2^{- dj} \log \Xi_{[0,2^j)^d}(z) = \sum_{j \in \Z} 2^{- dj} \log(1 + \widehat z_j)
\]
and the \emph{stability threshold}
\[
	\theta^* : = \limsup_{j \to \infty} 2^{- dj} \log z_j
\]
are finite with $\theta^* \leq p < \infty$, cf.\ \cite[Theorem~3.1]{jansen2020hierarchical}. (The proofs from~\cite{jansen2020hierarchical}, written for models without small scales, can be adapted to the present setup.)

For $B,B'\in \mathbb B$, let 
\[
	\mathrm{lsc}(B,B'):=  \min\{j \in \Z \mid \exists B'' \in \mathbb B_j : B'' \supset B, B'' \supset B' \}
\]
denote the lowest scale at which there is a block covering both $B$ and $B'$, and set
\[
	 D(B,B'):= 2^{d \, \mathrm{lsc}(B,B')}  \1_{\{B \neq B'\}}.
\]
We view $D(B,B')$ as an (ultra-)metric on the set of blocks and call it the \emph{hierarchical distance} (or metric). 

\begin{theorem} \label{thm:decay-hom-general}
	Suppose that $z$ admits a Gibbs measure $\P$. Then $\theta^* \leq p < \infty$ and, for all disjoint $B, B' \in \mathbb B$,
	\[
		\P(\omega \supset \{B, B'\}) - \P(\omega \ni B) \, \P(\omega \ni B') = \P(\omega \ni B) \, \P(\omega \ni B') \, R_{\mathrm{lcs}(B, B')}
	\]
	where $R_j = \prod_{k \geq j} (1 + \widehat z_k) - 1$, $j \in \mathbb Z$, satisfies $R_j \downarrow 0$ as $j \uparrow \infty$ as well as
	\[
		\limsup_{j \to \infty} 2^{- dj} \, \log(R_j) = \limsup_{j \to \infty} 2^{- dj} \, \log(\widehat z_j) = \theta^* - p \in [- \infty, 0].
	\]
\end{theorem}

\noindent Note that the conclusions of the above Theorem~\ref{thm:decay-hom-general} in particular yield an exponential decorrelation of the form
\[
	|\P(\omega \supset \{B, B'\}) - \P(\omega \ni B) \, \P(\omega \ni B')| \leq C_{\mathrm{lcs}(B, B')}(\theta) \, \mathrm e^{- (p - \theta) \, D(B, B')}
\]
for all $\theta \in (\theta^*, p)$ and disjoint $B, B' \in \mathbb B$ with $\lim_{j \to \infty} C_j(\theta) \to 0$. Assuming now $\theta^* > - \infty$, for $\theta = \theta^*$, i.e.\ for the supremal rate $p - \theta^* \geq 0$, we can only guarantee
\[
	\limsup_{j \to \infty} 2^{- dj} \, \log(C_j(\theta^*)) = 0
\]
but we have no generally available bound on $\limsup_{j \to \infty} C_j(\theta^*)$.

Exponential decay of correlations usually involves some Euclidean distance between lattice sites (for spin systems) or, in stochastic geometry, a Hausdorff distance on the set of non-empty compact subsets of $\R^d$ (``particles'' or ``grains''), cf.\ \cite[Subsection~3.3]{benesetal2020decorrelation}. Our theorem instead works with the hierarchical distance. Note that if the Euclidean or Hausdorff metric between $B$ and $B'$ is large then so is the hierarchical distance,
\[
	2^{\mathrm{lcs}(B, B')} \geq \sup_{\mathbf x \in B, \mathbf x' \in B'} \max\{|x_k - x'_k| \mid k \in \{1, \ldots, d\}\},
\]
but the converse is false; at fixed (positive) Hausdorff distance, $\mathrm{lcs}(B, B')$ can be made arbitrarily large and, for prescribed $\mathrm{lcs}(B, B')$, the blocks can be chosen to have arbitrarily small Hausdorff distance.

The strict inequality $\theta^* < p$, implicitly needed for the described exponential decay, characterises the gas phase in the original paper, cf.\ \cite[Proposition~4.3]{jansen2020hierarchical} in particular. While the condensed phase is marked by the absence of Gibbs measures, the existence of a Gibbs measure with $\theta^* = p < \infty$ defines the coexistence region.

A situation that naturally lends itself to a refinement of the above analysis is the one encountered in \cite[Subsection~5.4]{jansen2020hierarchical} where the activity, tuned via a \emph{chemical potential}, yields a parametrised system possibly undergoing a \emph{phase transition} at a critical chemical potential that may or may not admit a Gibbs measure with $\theta^* = p < \infty$. For simplicity and concreteness, we henceforth only consider parametric homogeneous activities $z = z(\mu, J, \alpha)$ given by
\[
	z_j(\mu, J, \alpha) = \begin{cases}
		\exp \left( 2^{dj} \, \mu - 2^{\alpha dj} \, J \right) & \text{if } j \in \mathbb N_0, \\
		0 & \text{otherwise},
	\end{cases}
\]
with chemical potential $\mu \in \mathbb R$, coupling constant $J \in \mathbb R$ and fractional exponent $\alpha \in (0, 1)$. The fractional exponent reflects surface corrections and should be thought of as $\alpha = (d-1)/d$ (if $d\geq 2$), see \cite[Example 5.8]{jansen2020hierarchical}.

By \cite[Lemma~5.1(b)]{jansen2020hierarchical} and Theorem~\ref{thm:existence-hom}, if $J$ is sufficiently large, the system undergoes a first-order phase transition as $\mu$ increases. Precisely, the following holds true: There exists a unique $\mu_c = \mu_c(J, \alpha) \in \mathbb R \cup \{+\infty\}$ (necessarily positive in fact) such that
\begin{itemize}
	\item[(a)] for all $\mu < \mu_c$, $z$ admits a Gibbs measure with
	\[
		\mu = \theta^*(\mu, J, \alpha) = \theta^* < p = p(\mu, J, \alpha) < \infty,
	\]
	\item[(b)] for $\mu = \mu_c$, $z$ satisfies $\mu = \theta^* = p < \infty$,
	\item[(c)] for all $\mu > \mu_c$, $z$ satisfies $\mu = \theta^* = p < \infty$ and $\mathcal G(z) = \varnothing$.
\end{itemize}
Furthermore, there exist finite thresholds $J_1 = J_1(\alpha) \geq 0$ and $J_2 = J_2(\alpha) \geq J_1$ such that 
\begin{itemize}
	\item[(1)] for all $J < J_1$ (and for $J = 0$), $\mu_c = \infty$ (cf.\ \cite[Theorem~5.3, Corollary~5.4]{jansen2020hierarchical}),
	\item[(2)] for all $J \in (J_1, J_2)$, $\mu_c < \infty$ with $\mathcal G(z(\mu_c, J, \alpha)) = \varnothing$,
	\item[(3)] for all $J > J_2$, $\mu_c < \infty$ with $\mathcal G(z(\mu_c, J, \alpha)) \neq \varnothing$ (cf.\ \cite[Theorem~5.6, Corollary~5.7(b)]{jansen2020hierarchical}).
\end{itemize}
In case (2) the parametric system undergoes a \emph{continuous} phase transition while, in case (3), the phase transition is of \emph{first order} (see \cite[Section 5]{jansen2020hierarchical}). (We do not address the questions of whether or not $J_1 > 0$ or $J_2 > J_1$ and what happens at $J = J_1$ or $J = J_2$.)

\begin{theorem} \label{thm:decay-hom-parametric}
	Let $J \in \mathbb R$ and $\alpha \in (0, 1)$. Suppose that $\mu \leq \mu_c$ is such that $z$ admits a Gibbs measure $\mathbb P$. Then the conclusions of Theorem~\ref{thm:decay-hom-general} hold true and, if additionally $J > 0$, then
	\[
		\lim_{j \to \infty} \left( \log(R_j) + 2^{dj} \, (p - \theta^*) + 2^{\alpha dj} \, J \right) = 0.
	\]
\end{theorem}

\noindent If $J$ and $\alpha$ are such that $z$ admits a Gibbs measure at $\mu = \mu_c < \infty$ (in particular $J > 0$, see above), then, for disjoint $B,B'$,
\begin{align*}
	& |\P(\omega \supset \{B, B'\}) - \P(\omega \ni B) \P(\omega \ni B')| \\
	& \qquad \leq \begin{cases}
		C_{\mathrm{lcs}(B, B')}(\theta^*) \, \exp \left( - (p - \theta^*) \, D(B, B') \right) & \text{if } \mu < \mu_c, \\
		C'_{\mathrm{lcs}(B, B')} \, \exp \left( - J \, D(B, B')^\alpha \right) & \text{if } \mu = \mu_c, \\
	\end{cases}
\end{align*}
with $\lim_{j \to \infty} C_j(\theta^*) = 0$ and $\lim_{j \to \infty} C'_j = 1$. Note that the previously problematic supremal rate $p - \theta^*$ is now achieved with constants $C_j(\theta^*)$ that are not only bounded but exhibit an explicit subexponential decay.

The theorem illustrates that the decay of correlations at the point of phase transition can be slower than in the gas phase. In this example, we end up with stretched exponential decay of the order of $\exp( - \mathrm{const} D(B,B')^\alpha)$. 

\section{Proofs} \label{sec:proofs}

\noindent Here we prove our results. 
For the existence result, proven in Subsection~\ref{subsec:proofs-existence}, we first transform the GNZ Eq.~\eqref{eq:gnz2} into a one-sided conditional property along the hierarchy of $\mathbb B$ (Eq.~\eqref{eq:topdown}). Next we observe that Conditions~(i) and (ii) in Theorem~\ref{thm:existence} are indeed necessary and that there is only one candidate for a Gibbs measure: the hierarchical measure (defined by Eq.~\eqref{eq:hiero-subsets}). Plugging this measure back into the GNZ equation, we ultimately derive the sufficiency of Conditions~(i) and (ii) in Theorem~\ref{thm:existence}.

Propositions~\ref{prop:fragmentation} and \ref{prop:condensation} are subsequently proved in Subsection~\ref{subsec:proofs-cond-frag} by expressing the relevant finite-volume Gibbs states as hierarchical measures.

Finally, Subsection~\ref{subsec:proofs-decay} closes with the proofs of Theorems~\ref{thm:decay-hom-general} and \ref{thm:decay-hom-parametric}. By again writing Gibbs measures as hierarchical probabilities, the covariances of interest are easily factorised, after which purely analytical arguments reminiscent of \cite{jansen2020hierarchical} finish the job.

\subsection{Proof of Theorem~\ref{thm:existence}} \label{subsec:proofs-existence}

For the proofs, we introduce some additional notation. The set $\mathbb I_B$ of blocks intersecting $B$ is the union of the sets
\[
	\mathbb A_B : = \{B' \in \mathbb B \mid B' \supset B\} \quad \text{and} \quad \mathbb B_B : = \{B' \in \mathbb B \mid B' \subset B\}
\] 
of blocks above and, respectively, below $B$, including $B$ itself. Starred sets exclude the block $B$,
\[
	\mathbb I_B^* : = \mathbb I_B \setminus \{B\}, \quad \mathbb A_B^* : = \mathbb A_B \setminus \{B\}, \quad \mathbb B_B^* : = \mathbb B_B \setminus \{B\}.
\]
We extend the notation to subsets $\mathcal B\subset \mathbb B$ by
\[
		\mathbb I_\mathcal B : = \bigcup_{B\in\mathcal B} \mathbb I_B,\quad \mathbb I_\mathcal B^* : = \bigcup_{B\in\mathcal B} \mathbb I_B^*, \quad \text{ etc}. 
\] 
Note that the event of non-overlap between blocks can be written as
\[
	\Delta = \{\mathcal B \in \Omega \mid \mathcal B \cap \mathbb I_\mathcal B^* = \varnothing\} = \{\omega \cap \mathbb I_\mathcal B^* = \varnothing\}
\]
and that the same characterisation holds true when replacing $\mathbb I_\mathcal B^*$ by either $\mathbb A_\mathcal B^*$ or $\mathbb B_\mathcal B^*$. For $\zeta : \mathbb B \to \R_+$ and $\mathcal B\subset \mathbb B$ we set 
\[
	\zeta^\mathcal B : = \prod_{B\in \mathcal B} \zeta(B),
\] 
the empty product is $1$ and infinite products are defined as the obvious limits which are going to exist, possibly infinite, in all subsequent uses of this notation.

Next we record a few simple facts. For $\mathbb P\in \mathcal G(z)$, the univariate Eq.~\eqref{eq:gnz2} immediately implies the bound $\P(\omega \ni B) \leq z(B)$ and, via a straightforward induction, also the multivariate equation, cf.\ \cite[Lemma~2.1]{betsch2023gibbspointprocess},
\begin{equation} \label{eq:gnz3}
	\P(\omega \supset \mathcal B, \, \omega \setminus \mathcal B \in \mathcal A) = \1_\Delta(\mathcal B) \, z^\mathcal B \, \P(\omega\cap \mathbb I_{\mathcal B} = \varnothing,\, \omega \setminus \mathcal B \in \mathcal A),
\end{equation} 
valid for all finite subsets $\mathcal B\subset \mathbb B$ and all measurable $\mathcal A\subset \Omega$. 

The partition functions satisfy a recurrence relation, cf.\ \cite[Eq.~(3.1)]{jansen2020hierarchical}: for $\Lambda \in \mathbb B_j$,
\begin{equation} \label{eq:pfr}
	\Xi_\Lambda(z) = z(\Lambda) + \prod_{\substack{\Lambda'\in \mathbb B_{j-1}: \\ \Lambda'\subset \Lambda}} \Xi_{\Lambda'}(z) = \bigl(1 + \widehat z(\Lambda)\bigr) \prod_{\substack{\Lambda'\in \mathbb B_{j-1}: \\ \Lambda'\subset \Lambda}} \Xi_{\Lambda'}(z). 
\end{equation} 
The products are over the $2^d$ subblocks of $\Lambda$ that are exactly one scale below $\Lambda$. Set
\[
	\widehat \rho_z(B):= \frac{\widehat z(B)}{1+\widehat z(B)},
\]
cf.\ \cite[Lemma~4.5(i)]{jansen2020hierarchical}, and note that the recurrence relation~\eqref{eq:pfr} yields the alternative expression
\begin{equation} \label{eq:rhohat2}
	\widehat \rho_z(B) =\frac{z(B)}{\Xi_B(z)}. 
\end{equation}

Gibbs measures satisfy a property that reflects a top-down construction implicit in the proof of \cite[Theorem~3.2]{jansen2020hierarchical}: Conditional on the configuration at scales $k\geq j+1$, the configuration at scale $j$ is that of blocks placed independently with probabilities $\widehat \rho_z(B)$ or zero, depending on the compatibility with higher-up scales. We formulate a slightly different version that relates the conditional probability of seeing a fixed block, given the state of all blocks that do not lie beneath it. (Cf.\ also the definition of $\P_p$ in Proposition~\ref{prop:mandelbrot}.)

Given $\widehat \rho:\mathbb B\to [0,1]$, we say that a measure $\mathbb P$ satisfies the \emph{top-down condition} for $\widehat \rho$ if 
\begin{equation} \label{eq:topdown}
	\P(\omega \ni B, \, \omega \setminus \mathbb B_B \in \mathcal A) = \widehat \rho(B) \, \P(\omega \cap \mathbb A_B^* = \varnothing, \, \omega \setminus \mathbb B_B \in \mathcal A)
\end{equation}
for all $B \in \mathbb B$ and all measurable $\mathcal A \subset \Omega$.

Comparing Eq.~\eqref{eq:topdown} to Eq.~\eqref{eq:gnz2}, the former seems to look like a version of the latter where the void probability on $\mathbb B_B = \mathbb I_B \setminus \mathbb A_B^*$ was somehow absorbed into the prefactor. Let us make rigorous sense out of this intuition.

\begin{lemma} \label{lem:topdown} 
	Every $\P \in \mathcal G(z)$ satisfies the top-down condition for $\widehat \rho_z$. 
\end{lemma} 

\begin{proof}
	Let $\P \in \mathcal G(z)$, $B \in \mathbb B$ and $\mathcal A \subset \Omega$ measurable. By the GNZ Eq.~\eqref{eq:gnz2},
	\[
		\P(\omega \ni B, \, \omega \setminus \mathbb B_B \in \mathcal A) = z(B) \, \P(\omega \cap \mathbb I_B = \varnothing, \, \omega \setminus \mathbb B_B \in \mathcal A)
	\]
	so, using Eq.~\eqref{eq:rhohat2}, it suffices to prove
	\begin{equation}  \label{eq:h1}
		\P(\omega \cap \mathbb I_B = \varnothing, \, \omega \setminus \mathbb B_B \in \mathcal A) = \frac{1}{\Xi_B(z)} \, \P(\omega \cap \mathbb A_B^* = \varnothing, \, \omega \setminus \mathbb B_B \in \mathcal A).
	\end{equation}
	Essentially choosing $\mathcal A = \Omega$, we henceforth omit the specification $\omega \setminus \mathbb B_B \in \mathcal A$ for better readability but the same argument goes through with the latter event reinserted everywhere. Write
	\[
		\P \bigl( \omega \cap \mathbb A_B^* = \varnothing, \, |\omega \cap \mathbb B_B| < \infty \bigr) = \sum_{ \substack{ \mathcal B\subset \mathbb B_B: \\ |\mathcal B| < \infty } } \P \bigl( \omega \cap \mathbb A_B^* = \varnothing, \, \omega \cap \mathbb B_B = \mathcal B \bigr)
	\]
	and use the multivariate GNZ Eq.~\eqref{eq:gnz3} to express each summand as
	\begin{align*}
		& \P \bigl (\omega \cap \mathbb A_B^* = \varnothing, \, \omega \cap \mathbb B_B = \mathcal B \bigr) \\
		& \qquad = \1_\Delta(\mathcal B) \, z^\mathcal B \, \P \bigl( \omega \cap \mathbb A_B^* = \omega \cap \mathbb I_\mathcal B = \omega \cap \mathbb B_B \setminus \mathcal B = \varnothing \bigr) \\
		& \qquad = \1_\Delta(\mathcal B) \, z^\mathcal B \, \P(\omega \cap \mathbb I_B = \varnothing). 
	\end{align*}
	Thus, we get
	\begin{equation} \label{eq:h2}
		\P \bigl( \omega \cap \mathbb A_B^* = \varnothing, \, |\omega \cap \mathbb B_B| < \infty \bigr) = \Xi_B(z) \, \P(\omega \cap \mathbb I_B = \varnothing).
	\end{equation}
	If $\Xi_B(z)$ is infinite, then the probability on the right side must be zero because, otherwise, the probability on the left would be infinite. Thus, both sides in Eq.~\eqref{eq:h1} vanish and the equation holds true. 
	
	If $\Xi_B(z)$ is finite, then we may omit the requirement $|\omega \cap \mathbb B_B| < \infty$ on the left side in Eq.~\eqref{eq:h2}. Indeed, as seen in Eq.~\eqref{eq:xifi-zfi}, $\Xi_B(z) < \infty$ implies $\sum_{B' \in \mathbb B_B} z(B') < \infty$ and hence also
	\[
		\sum_{B' \in \mathbb B_B} \P(\omega \ni B') < \infty.
	\]
	By the Borel-Cantelli lemma, said finiteness condition holds $\P$-almost surely. 
	Therefore, we may omit it from Eq.~\eqref{eq:h2}, then divide by the (finite) partition function on both sides and obtain Eq.~\eqref{eq:h1}.
\end{proof}

\begin{remark}
	Note that, in the case $\Xi_B(z) = \infty$, Eq.~\eqref{eq:h2} actually necessitates
	\[
		\P \bigl( \omega \cap \mathbb A_B^* = \varnothing, \, |\omega \cap \mathbb B_B| = \infty \bigr) = \P \bigl( \omega \cap \mathbb A_B^* = \varnothing \bigl).
	\]
	and, if the right-hand side is positive, $\sum_{B' \in \mathbb B_B} \P(\omega \ni B') = \infty$ by the Borel--Cantelli lemma. If $\P \in \mathcal G(z)$, the latter sum is bounded from above by
	\[
		\sum_{B' \in \mathbb B_B} \widehat \rho_z(B') \leq \sum_{B' \in \mathbb B_B} \widehat z(B') \leq \sum_{\substack{B' \in \mathbb B_B : \\ \Xi_{B'}(z) < \infty}} z(B'),
	\]
	where we used Lemma~\ref{lem:topdown} and Eq.~\eqref{eq:topdown}, followed by the definitions of $\widehat \rho_z$ and $\widehat z$. In view of the upcoming Lemma~\ref{lem:rhohat-summable}, applied to $\widehat \rho = \widehat \rho_z < 1$, we see that both Conditions~(i) and (ii) in Theorem~\ref{thm:existence} are indeed necessary for $\mathcal G(z) \neq \varnothing$.
\end{remark}

\noindent Naturally, the next question is that of existence of solutions to Eq.~\eqref{eq:topdown}. We start by deducing necessary conditions.

\begin{lemma} \label{lem:rhohat-summable}
	Let $\P$ satisfy the top-down condition for $\widehat \rho:\mathbb B \to[0,1]$.
	Then
	\[
		\P(\omega \cap \mathbb A_\mathcal B = \varnothing) = (1 - \widehat \rho)^{\mathbb A_\mathcal B}
	\]
	and $\sum_{B \in \mathbb A_\mathcal B} \widehat \rho(B) < \infty$ for all finite $ \mathcal B \subset \mathbb B$.
\end{lemma} 

\begin{proof}
	Fix a finite $\mathcal B \subset \mathbb B$. Without loss, assume that $\mathcal B \in \Delta$ since one may anyway pass to
	\[
		\min(\mathcal B) := \{B \in \mathcal B \mid \forall B' \in \mathcal B: B' \subset B \Leftrightarrow B' = B\} = \mathcal B \setminus \mathbb A_\mathcal B^* \in \Delta,
	\]
	the subset of the minimal cubes in $\mathcal B$ with respect to inclusion.
	Now, for any $B \in \mathcal B = \min(\mathcal B)$, one has $\mathbb A_\mathcal B \cap \mathbb B_B = \{B\}$ and Eq.~\eqref{eq:topdown} yields
	\begin{align*}
		\P(\omega \cap \mathbb A_\mathcal B = \varnothing) & = \P(\omega \cap \mathbb A_\mathcal B \setminus \{B\} = \varnothing) - \P(\omega \ni B, \, \omega \cap \mathbb A_\mathcal B \setminus \{B\} = \varnothing) \\
		& = (1 - \widehat \rho(B)) \, \P(\omega \cap \mathbb A_\mathcal B \setminus \{B\} = \varnothing) \\
		& = (1 - \widehat \rho(B)) \, \P(\omega \cap \mathbb A_{\mathcal B'} = \varnothing)
	\end{align*}
	with $\mathcal B' = \min(\mathbb A_\mathcal B \setminus \{B\})$ differing from $\mathcal B \setminus \{B\}$ by at most the unique block in $\min(\mathbb A_B^*)$. Inductively, we obtain
	\[
		\P(\omega \cap \mathbb A_\mathcal B = \varnothing) = (1 - \widehat \rho)^{\mathbb A_\mathcal B \setminus \mathbb A_B} \, \P(\omega \cap \mathbb A_B = \varnothing)
	\]
	for all $B \in \mathbb A_\mathcal B$ since $\mathbb A_\mathcal B \setminus \mathbb A_B$ is finite.
	
	Hence, it remains to consider the case $\mathcal B = \{B\}$ for a fixed block $B \in \mathbb B_j$. Denote by $B_n$ the unique element of $\mathbb A_B \cap \mathbb B_{j + n}$ for every $n \in \N_0$. Then we have
	\[
		\P(\omega \cap \mathbb A_{B_m} = \varnothing) = (1 - \widehat \rho)^{\mathbb A_{B_m} \setminus \mathbb A_{B_n}} \, \P(\omega \cap \mathbb A_{B_n} = \varnothing)
	\]
	for all $m, n \in \N$ with $m \leq n$ and
	\[
		\lim_{n \to \infty} \P(\omega \cap \mathbb A_{B_n} = \varnothing) = 1
	\]
	because the (increasing) union of the above events contains the event
	\[
		\bigcap_{B' \in \mathbb B} \{\omega \ni B' \Rightarrow \omega \cap \mathbb A_{B'}^* = \varnothing\} = \{\omega \cap \mathbb A_\omega^* = \varnothing\} = \Delta
	\]
	of $\P$-measure one, cf.\ Eq.~\eqref{eq:topdown}. Thus,
	\[
		\P(\omega \cap \mathbb A_B = \varnothing) = (1 - \widehat \rho)^{\mathbb A_B}
	\]
	and $\lim_{n \to \infty} (1 - \widehat \rho)^{\mathbb A_{B_n}} = 1$, yielding
	\[
		\sum_{B' \in \mathbb A_B} \widehat \rho(B') = \sum_{n \in \N_0} \widehat \rho(B_n) < \infty.
	\]
	The lemma follows from combining the latter results with the initial argument for arbitrary finite $\mathcal B \subset \mathbb B$.
\end{proof}

\noindent Combining Lemmas~\ref{lem:topdown} and \ref{lem:rhohat-summable} and the remark in between, we see that, for any Gibbs measure to exist, $\widehat \rho_z$ must be summable over all sets $\mathbb A_B$, $B \in \mathbb B$, which is in fact equivalent to Condition~(ii) in Theorem~\ref{thm:existence}, and Condition~(i) must also hold true. Furthermore, the first part of the latter lemma actually limits our options for candidate Gibbs states down to one of the following type.

\begin{lemma} \label{lem:hierarchical}
	Let $\widehat \rho:\mathbb B\to [0,1]$. There is a uniquely defined probability measure $\mathbb H_{\widehat \rho}$ such that 
	\begin{equation} \label{eq:hiero-subsets}
		\mathbb H_{\widehat \rho} \bigl( \omega \supset \mathcal B\bigr) = \1_{\Delta}(\mathcal B) \,  \widehat \rho^\mathcal B \, (1- \widehat \rho)^{\mathbb A_\mathcal B^*}
	\end{equation}
	for all finite sets $\mathcal B\subset \mathbb B$.
\end{lemma}

\noindent We call the measure $\mathbb H_{\widehat \rho}$ the \emph{hierarchical measure} associated with $\widehat \rho$. 

\begin{proof} 
	First off, Eq.~\eqref{eq:hiero-subsets} determines the probabilities of the events of the form $\{\omega \supset \mathcal B\}$ where $\mathcal B$ runs through the finite subsets of $\mathbb B$. These events form a generating $\pi$-system of the $\sigma$-field of $\Omega$; the $\pi$--$\lambda$ theorem implies uniqueness.

	Let $\P_{\widehat \rho}^\mathrm{Ber}$ be the probability measure on $\Omega$ under which the occupation numbers $\mathcal B \mapsto \1_{\{\mathcal B \ni B\}}$ are independent Bernoulli random variables with 
	\[
		\P_{\widehat \rho}^\mathrm{Ber}(\omega\ni B) = \widehat \rho(B). 
	\] 
	The Bernoulli measure $\P_{\widehat \rho}^\mathrm{Ber}$ allows for configurations with overlapping blocks but we keep only some of them: For $\mathcal B \subset \mathbb B$, let
	\[
		\max(\mathcal B) := \{B \in \mathcal B \mid \forall B' \in \mathcal B: B' \supset B \Leftrightarrow B' = B\} = \mathcal B \setminus \mathbb B_\mathcal B^* \in \Delta
	\]
	be the set of cubes that are maximal in $\mathcal B$ with respect to the partial order of inclusion. Then
	\[
		\P_{\widehat \rho}^\mathrm{Ber} \bigl( \max(\omega) \supset \mathcal B \bigr) = \P_{\widehat \rho}^\mathrm{Ber} \bigl( \omega \supset \mathcal B, \, \omega \cap \mathbb A_\mathcal B^* = \varnothing \bigr) = \1_\Delta(\mathcal B) \, \widehat \rho^\mathcal B \, (1 - \widehat \rho)^{\mathbb A_\mathcal B^*} 
	\]
	for all finite $\mathcal B \subset \mathbb B$. Therefore the image of $\P_{\widehat \rho}^\mathrm{Ber}$ under the map $\mathcal B \mapsto \max(\mathcal B)$ is the unique choice for $\mathbb H_{\widehat \rho}$ that satisfies Eq.~\eqref{eq:hiero-subsets}.
\end{proof}

\begin{remark}
	Note that the measure $\P_p$ from Prop.~\ref{prop:mandelbrot} is just the hierarchical measure associated with $\widehat \rho = p \, \1_{\mathbb B_{[0, 1)^d}}$.
\end{remark}

\noindent If any measure satisfies the top-down condition, it is the hierarchical one.

\begin{lemma} \label{lem:topdown-sol}
	Let $\widehat \rho: \mathbb B \to [0,1]$ with $\sum_{B' \in \mathbb A_B} \widehat \rho(B') < \infty$ for all $B \in \mathbb B$. Then $\mathbb H_{\widehat \rho}$ is the unique solution to the top-down equations for $\widehat \rho$. 
\end{lemma} 

\begin{proof}
	We begin by showing that $\mathbb H_{\widehat \rho}$, by its construction in the proof of Lemma~\ref{lem:hierarchical}, indeed satisfies the top-down condition. For fixed $B \in \mathbb B$ and $\mathcal A \subset \Omega$ measurable, we first note that, for all $\mathcal B \subset \mathbb B$, the set
	\[
		\max(\mathcal B) \setminus \mathbb B_B = \mathcal B \setminus \bigl( \mathbb B_\mathcal B^* \cup \mathbb B_B \bigr) = \bigl( \mathcal B \setminus \{B\} \bigr) \setminus \bigl( \mathbb B_{\mathcal B \setminus \{B\}}^* \cup \mathbb B_B^* \bigr)
	\]
	is unaffected by whether or not $B \in \mathcal B$. Therefore, the left-hand side of Eq.~\eqref{eq:topdown} for $\P = \mathbb H_{\widehat \rho}$ reads
	\begin{align*}
		\mathbb H_{\widehat \rho}(\omega \ni B, \, \omega \setminus \mathbb B_B \in \mathcal A) & = \P_{\widehat \rho}^\mathrm{Ber}(\omega \ni B, \, \omega \cap \mathbb A_B^* = \varnothing, \, \max(\omega) \setminus \mathbb B_B \in \mathcal A) \\
		& = \widehat \rho(B) \, P_{\widehat \rho}^\mathrm{Ber}(\omega \cap \mathbb A_B^* = \varnothing, \, \max(\omega) \setminus \mathbb B_B \in \mathcal A).
	\end{align*}
	On the other hand, for any $\mathcal B \subset \mathbb B$, we have $\max(\mathcal B) \cap \mathbb A_B^* = \varnothing$ if and only if $\mathcal B \cap \mathbb A_B^*$ is either empty or infinite. Since $\sum_{B' \in \mathbb A_B^*} \widehat \rho(B') < \infty$, the Borel-Cantelli lemma yields
	\[
		\mathbb H_{\widehat \rho}(\omega \cap \mathbb A_B^* = \varnothing, \, \omega \setminus \mathbb B_B \in \mathcal A) = \P_{\widehat \rho}^\mathrm{Ber}(\omega \cap \mathbb A_B^* = \varnothing, \, \max(\omega) \setminus \mathbb B_B \in \mathcal A).
	\]
	Hence, $\P = \mathbb H_{\widehat \rho}$ satisfies Eq.~\eqref{eq:topdown}.
	
	Conversely, let $\P$ solve the top-down equations for $\widehat \rho$. As noted in the proof of Lemma~\ref{lem:rhohat-summable}, $\P(\Delta) = 1$ and, for every finite $\mathcal B \in \Delta$ and any $B \in \mathcal B$,
	\[
		\P(\omega \supset \mathcal B) = \widehat \rho(B) \, \P(\omega \cap \mathbb A_B^* = \varnothing, \, \omega \supset \mathcal B \setminus \{B\}).
	\]
	By induction,
	\[
		\P(\omega \supset \mathcal B) = \1_\Delta(\mathcal B) \, \widehat \rho^\mathcal B \, \P(\omega \cap \mathbb A_\mathcal B^* = \varnothing)
	\]
	for all finite $\mathcal B \subset \mathbb B$ and, by Lemma~\ref{lem:rhohat-summable}, the latter expression equals the right-hand side of Eq.~\eqref{eq:hiero-subsets}. Lemma~\ref{lem:hierarchical} then yields $\P = \mathbb H_{\widehat \rho}$. 
\end{proof} 

\begin{remark}
	If $\sum_{B' \in \mathbb A_B} \widehat \rho(B') = \infty$ for some (and hence all choices of the) block $B \in \mathbb B$, the hierarchical measure is a Dirac on the empty configuration, $\mathbb H_{\widehat \rho} =\delta_{\varnothing}$, and it does not satisfy the top-down equations. Indeed, for $\P = \delta_{\varnothing}$ the top-down equations with $\mathcal A = \Omega$ yield $0 = \widehat \rho(B) \cdot 1$ for all $B \in \mathbb B$, contradicting the initial assumption.
\end{remark} 

\noindent By now, we know that, if $\mathcal G(z) \neq \varnothing$, then both Conditions~(i) and (ii) in Theorem~\ref{thm:existence} must be satisfied and $\mathcal G(z) = \{\mathbb H_{\widehat \rho_z}\}$. We now turn to the question of when a hierarchical measure is actually Gibbsian.

\begin{lemma} \label{lem:rhohat-to-z}
	Let $\widehat \rho: \mathbb B \to [0,1]$ with $\sum_{B' \in \mathbb A_B} \widehat \rho(B') < \infty$ for all $B \in \mathbb B$. Then $\mathbb H_{\widehat \rho} \in \mathcal G(z)$ if and only if
	\begin{equation} \label{eq:rhohat-to-z}
		\widehat \rho(B) = z(B) \, (1 - \widehat \rho)^{\mathbb B_B} 
	\end{equation} 
	for all $B\in \mathbb B$. 
\end{lemma} 

\begin{proof}
	In view of the proofs of Lemmas~\ref{lem:hierarchical}~and~\ref{lem:topdown-sol} and a standard extension argument via the $\pi$--$\lambda$ theorem, $\mathbb H_{\widehat \rho}$ satisfies the GNZ equation if and only if, for all $B \in \mathbb B$ and every finite $\mathcal B \subset \mathbb B \setminus \{B\}$, the terms
	\[
		\mathbb H_{\widehat \rho}(\omega \supset \{B\} \cup \mathcal B) = \1_\Delta(\{B\} \cup \mathcal B) \, \widehat \rho^{\{B\} \cup \mathcal B} \, (1 - \widehat \rho)^{\mathbb A_B^* \cup \mathbb A_\mathcal B^*}
	\]
	and
	\[
		z(B) \, \mathbb H_{\widehat \rho}(\omega \cap \mathbb I_B = \varnothing, \, \omega \supset \mathcal B) = z(B) \, \1_\Delta(\{B\} \cup \mathcal B) \, \widehat \rho^\mathcal B \, (1 - \widehat \rho)^{\mathbb I_B \cup \mathbb A_\mathcal B^*}
	\]
	coincide. Note that the conjunction of $\mathcal B \cap \mathbb I_B = \varnothing$ and $\mathcal B \in \Delta$ is indeed equivalent to $\{B\} \cup \mathcal B \in \Delta$, which also implies $\mathbb A_\mathcal B^* \setminus \mathbb A_B^* = \mathbb A_\mathcal B^* \setminus \mathbb I_B$. It hence suffices to compare the above terms in the case $\mathcal B = \varnothing$, that is, $\mathbb H_{\widehat \rho} \in \mathcal G(z)$ if and only if Eq.~\eqref{eq:rhohat-to-z} holds true for all $B\in \mathbb B$ such that
	\[
		(1 - \widehat \rho)^{\mathbb A_B^*} = \mathbb H_{\widehat \rho}(\omega \cap \mathbb A_B^* = \varnothing) > 0.
	\]
	Of course, for any given $B \in \mathbb B$, Eq.~\eqref{eq:rhohat-to-z} implies $\widehat \rho(B) < 1$ while the summability condition on $\widehat \rho$ yields the existence of some $B' \in \mathbb A_B^*$ with $(1 - \widehat \rho)^{\mathbb A_{B'}} > 0$. Combining these observations, one inductively obtains
	\[
		(1 - \widehat \rho)^{\mathbb A_B^*} = (1 - \widehat \rho)^{\mathbb A_{B'}} (1 - \widehat \rho)^{\mathbb A_B^* \setminus \mathbb A_{B'}} > 0
	\]
	whenever $\mathbb H_{\widehat \rho}$ is a Gibbs measure. The lemma now follows.
\end{proof}

\begin{remark}
	Note that Lemma~\ref{lem:rhohat-to-z} generalises Proposition~\ref{prop:mandelbrot}. The proof of the latter also proceeds by inserting the measure $\P = \P_p$ into the GNZ Eq.~\eqref{eq:gnz2} (with $B = [0, 1)^d$ and $\mathcal A = \Omega$). Since $\P_p = \mathbb H_{\widehat \rho}$ for $\widehat \rho = p \, \1_{\mathbb B_{[0, 1)^d}}$, the identity that fails in said proof is precisely Eq.~\eqref{eq:rhohat-to-z}.
\end{remark}

\noindent By plugging $\widehat \rho_z$ into Lemma~\ref{lem:rhohat-to-z}, we are set on the path to deriving the sufficiency of Conditions~(i) and (ii) in Theorem~\ref{thm:existence}.

\begin{lemma}\label{lem:pf-as-product}
	The map $\widehat \rho = \widehat \rho_z: \mathbb B \to [0,1]$ satisfies Eq.~\eqref{eq:rhohat-to-z} for all $B \in \mathbb B$ if and only if 
	\begin{equation}	\label{eq:pf-as-product}
		\Xi_B(z) = (1 + \widehat z)^{\mathbb B_B}
	\end{equation}
	for all $B\in \mathbb B$. 
\end{lemma} 

\begin{proof}
	By definition, $1 - \widehat \rho_z = 1 / (1 + \widehat z)$. Given any $B \in \mathbb B$, applying the former together with Eq.~\eqref{eq:rhohat2} turns Eq.~\eqref{eq:rhohat-to-z} into
	\[
		 \frac{z(B)}{\Xi_B(z)} = \frac{z(B)}{(1 + \widehat z)^{\mathbb B_B}},
	\]
	which certainly holds true whenever Eq.~\eqref{eq:pf-as-product} does.
	
	Conversely, suppose that Eq.~\eqref{eq:rhohat-to-z} holds true for all $B \in \mathbb B$. Fix some $\Lambda \in \mathbb B$ and observe that Eq.~\eqref{eq:rhohat-to-z} still holds true for all $B \in \mathbb B$ when $z$, $\widehat z$ and $\widehat \rho = \widehat \rho_z$ are replaced by
	\[
		z_\Lambda = z \, \1_{\mathbb B_\Lambda}, \quad \widehat{(z_\Lambda)} = \widehat z \, \1_{\mathbb B_\Lambda} \quad \text{and} \quad \widehat \rho_{z_\Lambda} = \widehat \rho_z \, \1_{\mathbb B_\Lambda},
	\]
	respectively. Hence, given any $B \in \mathbb B$, we may assume without loss that $\sum_{B' \in \mathbb A_B^*} \widehat \rho_z(B') = 0$ and that, by Lemma~\ref{lem:rhohat-to-z}, $\mathbb H_{\widehat \rho_z}$ is a Gibbs measure for which Eq.~\eqref{eq:h1} with
	$\mathcal A = \Omega$ reads
	\[
		\frac{1}{(1 + \widehat z)^{\mathbb B_B}} = \frac{1}{\Xi_B(z)}
	\]
	since $\widehat z$ and $\widehat \rho_z$ are assumed to vanish on $\mathbb A_B^*$. The latter identity is obviously just Eq.~\eqref{eq:pf-as-product}.
\end{proof}

\noindent Cubes with finite partition function always satisfy Condition~\eqref{eq:pf-as-product}. 

\begin{lemma} \label{lem:finite-pf}
	For all $\Lambda \in \mathbb B$ with $\Xi_\Lambda(z) < \infty$, one has $\Xi_\Lambda(z) = (1 + \widehat z)^{\mathbb B_\Lambda}$.
\end{lemma}

\begin{proof}
	For $n \in \Z$, set
	\[
	\mathbb B^{(n)} : = \bigcup_{j \geq -n} \mathbb B_j, \quad \mathbb B_\Lambda^{(n)} : = \mathbb B_\Lambda \cap \mathbb B^{(n)}.
	\]
	By the recurrence relation \eqref{eq:pfr} for partition functions, for every $n \in \Z$ with $-n-1$ at most the scale of $\Lambda$, 
	\[
		\Xi_\Lambda(z) = (1 + \widehat z)^{\mathbb B_\Lambda^{(n)}}  \prod_{B \in \mathbb B_\Lambda \cap \mathbb B_{-n-1}} \Xi_B(z). 
	\] 	
	When $z$ vanishes for blocks at scale below $-n$, the partition functions on the right side are equal to $1$ and can be omitted from the previous equation. Thus, let us consider the truncated activity 
	\[
		z^{(n)}(B):= \begin{cases} 
						z(B) & \text{if } B \in \mathbb B^{(n)}, \\
						0 & \text{else},
				\end{cases}
	\] 	
	and write $\widehat z^{(n)}(B)$ for the associated effective activity. Then 
	\begin{equation} \label{eq:xifi}
		\Xi_\Lambda \bigl( z^{(n)} \bigr) = \bigl( 1 + \widehat z^{(n)} \bigr)^{\mathbb B_\Lambda}
	\end{equation}
	and only finitely many terms on the right side, namely those for $B \in \mathbb B_\Lambda^{(n)}$, may differ from $1$. The left side converges monotonically to $\Xi_\Lambda(z)$ as $n\to \infty$. For the right side, we note that $\widehat z^{(n)}(B) \to \widehat z(B)$ for each $B \in \mathbb B$ and 
	$\widehat z^{(n)}(B)  \leq z^{(n)}(B) \leq z(B)$ with
	\[
		\sum_{B \in \mathbb B_\Lambda} \log \bigl( 1 + z(B) \bigr) \leq \sum_{B \in \mathbb B_\Lambda} z(B) \leq \Xi_\Lambda(z) < \infty. 
	\] 
	Dominated convergence allows us to pass to the limit on the right side of Eq.~\eqref{eq:xifi}.
\end{proof}

\noindent This leaves blocks with infinite partition function, for which Eq.~\eqref{eq:pf-as-product} turns out to be equivalent to Condition~(i) in Theorem~\ref{thm:existence}.

\begin{lemma} \label{lem:infinite-pf}
	If $\Lambda \in \mathbb B$ such that $\Xi_\Lambda(z) = \infty$, then $\Xi_\Lambda(z) = (1 + \widehat z)^{\mathbb B_\Lambda}$ holds true if and only if $\sum_{B \in \mathbb B_\Lambda: \Xi_B(z) < \infty} z(B) = \infty$. 	
\end{lemma} 

\begin{proof}
 If $\Xi_\Lambda(z) = (1 + \widehat z)^{\mathbb B_\Lambda} = \infty$, then we must have 
	\[
		\sum_{\substack{ B \in \mathbb B_\Lambda: \\ \Xi_B(z) < \infty }} z(B) \geq \sum_{B \in \mathbb B_\Lambda} \widehat z(B) = \infty.  
	\] 
	Conversely, suppose that $\sum_{B \in \mathbb B_\Lambda: \Xi_B(z) < \infty} z(B) = \infty$ holds true. Let 
	\[
		\mathbb F_\Lambda(z) := \{B \in \mathbb B_\Lambda \mid \Xi_B(z) < \infty\} \quad \text{and} \quad \mathbb M_\Lambda(z) := \max(\mathbb F_\Lambda(z))
	\] 	
	 be the relevant set of blocks below $\Lambda$ with finite partition function and the set of its maximal elements with respect to inclusion, respectively, cf.\ the proof of Lemma~\ref{lem:hierarchical}. Writing
	 \[
	 	\mathbb F_\Lambda(z) = \mathbb B_{\mathbb M_\Lambda(z)} = \bigcup_{M \in \mathbb M_\Lambda(z)} \mathbb B_M, 
	 \] 
	 where the last union is over mutually disjoint sets of cubes, we obtain
	 \[
	 	(1 + \widehat z)^{\mathbb B_\Lambda} = (1 + \widehat z)^{\mathbb F_\Lambda(z)} = \prod_{M \in \mathbb M_\Lambda(z)} (1 + \widehat z)^{\mathbb B_M} = \prod_{M \in \mathbb M_\Lambda(z)} \Xi_M(z). 
	 \] 
	 In the last line we have applied Lemma~\ref{lem:finite-pf} to the blocks $M \in \mathbb M_\Lambda(z)$. The product of partition functions is bounded from below by 
	 \[
	 	\prod_{M \in \mathbb M_\Lambda(z)} \Bigl( 1 + \sum_{B  \in \mathbb B_M} z(B) \Bigr) \geq 1 + \sum_{B\in \mathbb F_\Lambda(z)} z(B) = \infty. 
	 \]  
	 Thus $\prod_{M \in \mathbb M_\Lambda(z)} \Xi_M(z)$ is infinite and so is $(1+\widehat z)^{\mathbb B_\Lambda}$.
\end{proof}

\noindent At this point, the proof of our main result on the existence and uniqueness of Gibbs measures comes down to collecting the partial results derived so far.

\begin{proof}[Proof of Theorem~\ref{thm:existence}]
	By Lemmas~\ref{lem:topdown},~\ref{lem:rhohat-summable} and~\ref{lem:topdown-sol}, any $\P \in \mathcal G(z)$ satisfies the top-down condition for $\widehat \rho_z$ and is hence equal to $\mathbb H_{\widehat \rho_z}$ with
	\[
		\sum_{B' \in \mathbb A_B} \widehat \rho_z(B') < \infty
	\]
	for all $B \in \mathbb B$. This covers both the uniqueness of the Gibbs measure as well as the necessity of Condition~(ii) in Theorem~\ref{thm:existence}, the latter being equivalent to the above summability assertion for $\widehat \rho_z = \widehat z / (1 + \widehat z)$.
	
	As previously noted, the necessity of Condition~(i) in Theorem~\ref{thm:existence} also follows from the arguments proving Lemmas~\ref{lem:topdown} and \ref{lem:rhohat-summable}. Nevertheless, it now suffices to prove the following: given that the above summability condition on $\widehat \rho_z$ (i.e.\ Condition~(ii) in Theorem~\ref{thm:existence}) holds true, one has $\mathbb H_{\widehat \rho_z} \in \mathcal G(z)$ if and only if Condition~(i) in Theorem~\ref{thm:existence} is met. But this is essentially the content of Lemmas~\ref{lem:rhohat-to-z} through \ref{lem:infinite-pf}.
\end{proof}

\noindent This concludes the proof of our main result. The explicit representation of Gibbs states as hierarchical measures is of interest in its own right and very useful for subsequent proofs.

\subsection{Proofs of Propositions~\ref{prop:fragmentation} and \ref{prop:condensation}} \label{subsec:proofs-cond-frag}

Recall that, given an activity $z: \mathbb B \to \R_+$ and $\Lambda \in \mathbb B$, we denote by $z_\Lambda = z \, \1_{\mathbb B_\Lambda}$ the finite-volume restriction of $z$ to blocks below $\Lambda$ and by $\P_\Lambda \in \mathcal G(z_\Lambda)$ the corresponding (unique) finite-volume Gibbs-measure whenever it exists. If it does exist, then, by Lemmas~\ref{lem:topdown}, \ref{lem:rhohat-summable} and \ref{lem:topdown-sol}, it is just $\mathbb H_{\widehat \rho_{z_\Lambda}}$, the hierarchical measure associated with $\widehat \rho_{z_\Lambda} = \frac{\widehat z}{1 + \widehat z} \, \1_{\mathbb B_\Lambda}$, cf.\ also the proof of Lemma~\ref{lem:pf-as-product}.

For $n \in \Z$, recall the further restriction $z_\Lambda^{(n)} = z \, \1_{\mathbb B_\Lambda^{(n)}}$ of the activity to the finite set
\[
	\mathbb B_\Lambda^{(n)} := \mathbb B_\Lambda \cap \bigcup_{j \geq -n} \mathbb B_j,
\]
see also the proof of Lemma~\ref{lem:finite-pf}. Adding Lemmas~\ref{lem:rhohat-to-z} and \ref{lem:pf-as-product} to the mix, the corresponding Gibbs-measure $\P_\Lambda^{(n)}$ for $z_\Lambda^{(n)}$ exists unconditionally and is, of course, nothing but the hierarchical measure for $\widehat \rho_{z_\Lambda^{(n)}}$.

\begin{lemma} \label{lem:convergence-void-weak}
	One has
	\[
		\lim_{n \to \infty} \P_\Lambda^{(n)}(\omega = \varnothing) = \frac{1}{\Xi_\Lambda(z)} \quad \text{and} \quad \lim_{n \to \infty} \P_\Lambda^{(n)}(\omega \supset \mathcal B) = \mathbb H_{\widehat \rho_{z_\Lambda}}(\omega \supset \mathcal B)
	\]
	for all finite $\mathcal B \in \mathbb B$.
\end{lemma}

\begin{proof}
	Since, by the preceding discussion,
	\[
		\P_\Lambda^{(n)}(\omega = \varnothing) = \P_\Lambda^{(n)} \bigl( \omega \cap \mathbb B_\Lambda^{(n)} = \varnothing \bigr) = \frac{1}{\bigl( 1 + \widehat z_\Lambda^{(n)} \bigr)^{\mathbb B_\Lambda^{(n)}}} = \frac{1}{\Xi_\Lambda \bigl( z_\Lambda^{(n)} \bigr)}
	\]
	for all $n \in \Z$, the first statement follows by monotone convergence. As an immediate consequence,
	\[
		\lim_{n \to \infty} \widehat \rho_{z_\Lambda^{(n)}}(B) = \lim_{n \to \infty} \frac{z_\Lambda^{(n)}(B)}{\Xi_B \bigl( z_\Lambda^{(n)} \bigr)} = \frac{z_\Lambda(B)}{\Xi_B(z_\Lambda)} = \widehat \rho_{z_\Lambda}
	\]
	for all $B \in \mathbb B$ so we also get
	\begin{align*}
		\lim_{n \to \infty} \P_\Lambda^{(n)}(\omega \supset \mathcal B) & = \lim_{n \to \infty} \1_\Delta(\mathcal B) \, \widehat \rho_{z_\Lambda^{(n)}}^\mathcal B \, (1 - \widehat \rho_{z_\Lambda^{(n)}})^{\mathbb A_\mathcal B^*} \\
		& = \1_\Delta(\mathcal B) \, \widehat \rho_{z_\Lambda}^\mathcal B \, (1 - \widehat \rho_{z_\Lambda})^{\mathbb A_\mathcal B^*} \\
		& = \mathbb H_{\widehat \rho_{z_\Lambda}}(\omega \supset \mathcal B).
	\end{align*}
	for all finite $\mathcal B \in \mathbb B$ where we note that effectively only finitely many blocks, namely those in $\mathbb A_\mathcal B \cap \mathbb B_\Lambda$, contribute to the above products.
\end{proof}

\begin{remark}
	When equipping $\Omega \cong \{0, 1\}^\mathbb B$ with the product topology, with each factor viewed as discrete, the second part of Lemma~\ref{lem:convergence-void-weak} obtains a new meaning: it says that $\mathbb H_{\widehat \rho_{z_\Lambda}}$, the unique candidate for the finite-volume Gibbs measure $\P_\Lambda$, is always the weak limit of the truncated finite-volume Gibbs measures $\P_\Lambda^{(n)}$.
\end{remark}

\noindent Next, recall the set $\mathbb F_\Lambda(z) = \{B \in \mathbb B_\Lambda \mid \Xi_B(z) < \infty\}$ from the proof of Lemma~\ref{lem:infinite-pf}. The implications of the representation
\begin{equation} \label{eq:F-M-structure}
	\mathbb F_\Lambda(z) = \mathbb B_{\mathbb M_\Lambda(z)} = \bigcup_{M \in \mathbb M_\Lambda(z)} \mathbb B_M
\end{equation}
with $\mathbb M_\Lambda(z) = \max(\mathbb F_\Lambda(z))$ are particularly relevant to the last parts of the proof of the next lemma.

\begin{lemma} \label{lem:convergence-F-void}
	One has
	\[
		\lim_{n \to \infty} \P_\Lambda^{(n)}(\omega \cap \mathbb F_\Lambda(z) = \varnothing) = \frac{1}{(1 + \widehat z)^{\mathbb B_\Lambda}}.
	\]
	If $\Xi_\Lambda(z) = \infty$ and $\sum_{B \in \mathbb F_\Lambda(z)} z(B) < \infty$, one additionally has
	\[
		\lim_{n \to \infty} \P_\Lambda^{(n)}(\omega \setminus \mathbb F_\Lambda(z) = \varnothing) = 0.
	\]
\end{lemma}

\begin{proof}
	We prove the first statement in two steps. On the one hand, one trivially has
	\[
		\limsup_{n \to \infty} \P_\Lambda^{(n)}(\omega \cap \mathbb F_\Lambda(z) = \varnothing) \leq \inf_{m \in \Z} \limsup_{n \to \infty} \P_\Lambda^{(n)} \bigl( \omega \cap \mathbb F_\Lambda(z) \cap \mathbb B_\Lambda^{(m)} = \varnothing \bigr).
	\]
	In view of Lemma~\ref{lem:convergence-void-weak}, $\mathbb B_\Lambda^{(m)}$ being finite and $\widehat z$ vanishing on $\mathbb B_\Lambda \setminus \mathbb F_\Lambda(z)$, the probability in the upper bound actually converges to
	\[
		\lim_{n \to \infty} \P_\Lambda^{(n)} \bigl( \omega \cap \mathbb B_\Lambda^{(m)} = \varnothing \bigr) = \lim_{n \to \infty} \frac{1}{\bigl( 1 + \widehat z^{(n)} \bigr)^{\mathbb B_\Lambda^{(m)}}} = \frac{1}{(1 + \widehat z)^{\mathbb B_\Lambda^{(m)}}}
	\]
	whose infimum over $m \in \Z$ is just $1/(1 + \widehat z)^{\mathbb B_\Lambda}$.
	
	On the other hand, we obtain a lower bound via
	\begin{align*}
		\P_\Lambda^{(n)}(\omega \cap \mathbb F_\Lambda(z) = \varnothing) & = \sum_{ \substack{ \mathcal B \subset \mathbb B_\Lambda^{(n)} \setminus \mathbb F_\Lambda(z): \\ \mathcal B \in \Delta } } \widehat \rho_{z^{(n)}}^\mathcal B \, \bigl( 1 - \widehat \rho_{z^{(n)}} \bigr)^{\mathbb B_\Lambda^{(n)} \setminus \mathbb B_\mathcal B} \\
		& \geq \sum_{ \substack{ \mathcal B \subset \mathbb B_\Lambda^{(n)} \setminus \mathbb F_\Lambda(z): \\ \mathcal B \in \Delta } } \frac{\widehat \rho_{z^{(n)}}^\mathcal B \, \bigl( 1 - \widehat \rho_{z^{(n)}} \bigr)^{\mathbb B_\Lambda^{(n)} \setminus ( \mathbb B_\mathcal B \cup \mathbb F_\Lambda(z) )}}{\bigl( 1 + \widehat z^{(n)} \bigr)^{\mathbb F_\Lambda(z)}} \\
		& = \frac{1}{\bigl( 1 + \widehat z^{(n)} \bigr)^{\mathbb F_\Lambda(z)}}.
	\end{align*}
	While the summands in the first sum are hierarchical probabilities of the events $\{\omega \cap \mathbb B_\Lambda^{(n)} = \mathcal B\}$, the last equality comes about because, by Eq.~\eqref{eq:F-M-structure}, the numerators in the second sum are the corresponding hierarchical probabilities of the events $\{\omega \cap \mathbb B_\Lambda^{(n)} \setminus \mathbb F_\Lambda(z) = \mathcal B\}$. Lemmas~\ref{lem:finite-pf} and \ref{lem:infinite-pf}, together with monotone convergence, yield
	\[
		\bigl( 1 + \widehat z^{(n)} \bigr)^{\mathbb F_\Lambda(z)} = \Xi_\Lambda \bigl( z^{(n)} \, \1_{\mathbb F_\Lambda(z)} \bigr) \to \Xi_\Lambda \bigl( z \, \1_{\mathbb F_\Lambda(z)} \bigr) = (1 + \widehat z)^{\mathbb F_\Lambda(z)}
	\]
	as $n \to \infty$ (even if the limiting terms on the right are infinite) so we now observe
	\[
		\liminf_{n \to \infty} \P_\Lambda^{(n)}(\omega \cap \mathbb F_\Lambda(z) = \varnothing) \geq \frac{1}{(1 + \widehat z)^{\mathbb F_\Lambda(z)}} = \frac{1}{(1 + \widehat z)^{\mathbb B_\Lambda}},
	\]
	the term on the right coinciding with our upper bound.
	
	For the additional claim of the lemma, we write
	\[
		\P_\Lambda^{(n)}(\omega \setminus \mathbb F_\Lambda(z) = \varnothing) = \frac{1}{\bigl( 1 + \widehat z^{(n)} \bigr)^{\mathbb B_\Lambda \setminus \mathbb F_\Lambda(z)}} = \frac{\bigl( 1 + \widehat z^{(n)} \bigr)^{\mathbb F_\Lambda(z)}}{\bigl( 1 + \widehat z^{(n)} \bigr)^{\mathbb B_\Lambda}},
	\]
	cf.\ again Eq.~\eqref{eq:F-M-structure}. Building on previous arguments, the numerator and denominator on the right respectively converge to $(1 + \widehat z)^{\mathbb F_\Lambda(z)}$ and $\Xi_\Lambda(z)$ and the additional claim follows.
\end{proof}

\noindent Having completed our preparation, Proposition~\ref{prop:fragmentation} is easy to deal with.

\begin{proof}[Proof of Proposition~\ref{prop:fragmentation}]
	Unsurprisingly, we choose
	\begin{align*}
		\mathbb V(z) & = \left\{ \Lambda \in \mathbb B \, \Bigg\vert \, \Xi_\Lambda(z) = \infty, \sum_{B \in \mathbb F_\Lambda(z)} z(B) < \infty \right\} \\
		& = \bigl\{ \Lambda \in \mathbb B \mid \Xi_\Lambda(z) = \infty, \, (1 + \widehat z)^{\mathbb B_\Lambda(z)} < \infty \bigr\}
	\end{align*}
	and assume Condition~(i) in Theorem~\ref{thm:existence} to fail, i.e.\ $\mathbb V(z) \neq \varnothing$. Fix $\Lambda \in \mathbb V(z)$, observe that
	\[
		\mathbb B_\Lambda \cap \mathbb V(z) = \{B \in \mathbb B_\Lambda \mid \Xi_B(z) = \infty\} = \mathbb B_\Lambda \setminus \mathbb F_\Lambda(z)
	\]
	and that this intersection must be infinite by virtue of Eq.~\eqref{eq:pfr}, cf.\ also the beginning of the proof of Lemma~\ref{lem:finite-pf}.
	
	From Lemmas~\ref{lem:convergence-void-weak} and \ref{lem:convergence-F-void}, we already obtain
	\[
		\lim_{n \to \infty} \P_\Lambda^{(n)}(\omega \ni B) = \mathbb H_{\widehat \rho_{z_\Lambda}}(\omega \ni B) \leq \widehat \rho_z(B) = \frac{z(B)}{\Xi_B(z)} = 0
	\]
	for all $B \in \mathbb B_\Lambda \cap \mathbb V(z)$ and, specifically for $B = \Lambda$,
	\[
		\lim_{n \to \infty} \P_\Lambda^{(n)}(\omega \cap \mathbb B_B \neq \varnothing) = 1 - \frac{1}{\Xi_\Lambda(z)} = 1 = \lim_{n \to \infty} \P_\Lambda^{(n)}(\omega \cap \mathbb B_B \cap \mathbb V(z) \neq \varnothing)
	\]
	as well as
	\[
		\lim_{n \to \infty} \P_\Lambda^{(n)}(\omega \cap \mathbb B_B \setminus \mathbb V(z) \neq \varnothing) = 1 - \frac{1}{(1 + \widehat z)^{\mathbb B_\Lambda}} < 1.
	\]
	The latter restriction is easily overcome by writing
	\[
		\P_\Lambda^{(n)}(\omega \cap \mathbb B_B \cap \mathbb X \neq \varnothing) = \bigl( 1 - \widehat \rho_{z^{(n)}} \bigr)^{\mathbb A_B^* \cap \mathbb B_\Lambda} \, \P_B^{(n)}(\omega \cap \mathbb B_B \cap \mathbb X \neq \varnothing)
	\]
	with $\mathbb X \in \{\mathbb B, \mathbb V(z), \mathbb B \setminus \mathbb V(z)\}$ and noticing that the product over the finite set $\mathbb A_B^* \cap \mathbb B_\Lambda$ converges to $1$ whenever $B \in \mathbb B_\Lambda \cap \mathbb V(z)$.
\end{proof}

\noindent With fragmentation out of the way, we turn to condensation. There is no additional preparation needed but we reiterate that taking limits as $\Lambda \uparrow \R_+^d$ simply means taking limits as $\Lambda$ runs through an increasing sequence $\Lambda_1 \subset \Lambda_2 \subset \ldots$ of cubes in $\mathbb B$ with $\bigcup_{n \in \N} \Lambda = \R_+^d$.

\begin{proof}[Proof of Proposition~\ref{prop:condensation}]
	Suppose that $z$ satisfies Condition~(i) in Theorem~\ref{thm:existence}. Then, for any given $\Lambda \in \mathbb B$, $\P_\Lambda$ exists and is just the hierarchical measure associated with $\widehat \rho_{z_\Lambda} = \widehat \rho_z \, \1_{\mathbb B_\Lambda}$. In particular, one has
	\[
		\P_\Lambda(\omega \ni B) = \1_{\mathbb B_\Lambda}(B) \, \widehat \rho_z(B) \, (1 - \widehat \rho_z)^{\mathbb A_B^* \cap \mathbb B_\Lambda} = \1_{\mathbb B_\Lambda}(B) \, \frac{\widehat z(B)}{(1 + \widehat z)^{\mathbb A_B \cap \mathbb B_\Lambda}}
	\]
	and
	\[
		\P_\Lambda(\omega \cap \mathbb A_B \neq \varnothing) = 1 - (1 - \widehat \rho_z)^{\mathbb A_B \cap \mathbb B_\Lambda} = 1 - \frac{1}{(1 + \widehat z)^{\mathbb A_B \cap \mathbb B_\Lambda}}
	\]
	for all $\Lambda, B \in \mathbb B$, cf.\ also Lemma~\ref{lem:rhohat-summable}. Since Condition~(ii) in Theorem~\ref{thm:existence} can be equivalently phrased as
	\[
		\frac{1}{(1 + \widehat z)^{\mathbb A_B}} = \lim_{\Lambda \uparrow \R_+^d} \frac{1}{(1 + \widehat z)^{\mathbb A_B \cap \mathbb B_\Lambda}} > 0
	\]
	for all (or, equivalently, some) $B \in \mathbb B$, the proposition follows.
\end{proof}

\begin{remark}
	In analogy with the remark following the proof of Lemma~\ref{lem:convergence-void-weak}, it is easy to show that, provided their existence, the finite-volume Gibbs measures $\P_\Lambda$ always converge weakly to the unique candidate Gibbs state $\mathbb H_{\widehat \rho_z}$ as $\Lambda \uparrow \R_+^d$. The crux of the issue is yet again the extension of this convergence to appropriate void probabilities.
\end{remark}

\subsection{Proofs of Theorems~\ref{thm:decay-hom-general}~and~\ref{thm:decay-hom-parametric}} \label{subsec:proofs-decay}

Finally, let us tackle the exponential decay of block covariances. Although we formulated these results for homogeneous activities, the factorisation they are based on holds for all our Gibbs measures.

\begin{lemma} \label{lem:covariance-explicit}
	Suppose that $z$ admits a Gibbs measure $\mathbb P$. Then, for all distinct $B, B' \in \mathbb B$,
	\[
		\mathbb P(\omega \supset \{B, B'\}) = \1_\Delta(\{B, B'\}) \, \mathbb P(\omega \ni B) \, \mathbb P(\omega \ni B') \, (1 + R(B''))
	\]
	where $\{B''\} = \min(\mathbb A_B^* \cap \mathbb A_{B'}^*)$ and $R(B'') = (1 + \widehat z)^{\mathbb A_{B''}} - 1 \downarrow 0$ as $B'' \uparrow \mathbb R_+^d$.
\end{lemma}

\begin{proof}
	Fix two blocks $B, B' \in \mathbb B$. By the results of Subsection~\ref{subsec:proofs-existence}, $\P = \mathbb{H}_{\widehat \rho_z}$ and the relevant probabilities are given by
	\[
		\P(\omega \supset \{B, B'\}) = \1_\Delta(\{B, B'\}) \, \widehat \rho_z(B) \, \widehat \rho_z(B') \, (1 - \widehat \rho_z)^{\mathbb A_{\{B, B'\}}^*}
	\]
	as well as
	\[
		\P(\omega \ni \mathcal B) = \widehat \rho_z(B) \, (1 - \widehat \rho_z)^{\mathbb A_B^*} \quad \text{and} \quad \P(\omega \ni B') = \widehat \rho_z(B') \, (1 - \widehat \rho_z)^{\mathbb A_{B'}^*}.
	\]
	The result follows by inserting
	\[
		(1 - \widehat \rho_z)^\mathbb X = \frac{1}{(1 + \widehat z)^\mathbb X}, \quad \mathbb X \in \{\mathbb A_B^*, \, \mathbb A_{B'}^*, \, \mathbb A_{\{B, B'\}}^* = \mathbb A_B^* \cup \mathbb A_{B'}^*\},
	\]
	and observing that Condition~(ii) in Theorem~\ref{thm:existence} yields $(1 + \widehat z)^{\mathbb A_B} \downarrow 1$ as $B \uparrow \R_+^d$.
\end{proof}

\noindent What is left in the proofs of Theorems~\ref{thm:decay-hom-general} and \ref{thm:decay-hom-parametric} is mainly the asymptotic analysis of the factor $R(B)$ as $B \uparrow \mathbb R_+^d$. Clearly, $R(B) = (1 + \widehat z)^{\mathbb A_B} - 1$ has the same (multiplicative) asymptotics as $\sum_{B' \in \mathbb A_B} \widehat z(B')$ whenever either expression is finite but we can be far more precise than that. The following simple analytical observation will prove helpful.

\begin{lemma} \label{lem:gamma-series-summand-equivalence}
	Let $r > 0$ and $b > 1$. Then
	\[
	\e^{- r \, b^j} \leq \sum_{k \geq j} \e^{- r \, b^k} \leq \e^{- r \, b^j} \left( 1 + \frac{1}{r \, \log(b) \, b^j} \right)
	\]
	for all $j \in \mathbb Z$.
\end{lemma}

\begin{proof}
	Fix some arbitrary $j \in \mathbb Z$. Then
	\[
		\sum_{k \geq j + 1} \e^{- r \, b^k} \leq \int_j^\infty \frac{\dd x}{\e^{r \, b^x}} = \int_{r \, b^j}^\infty \frac{\dd y}{\log(b) \, y \, \e^y} \leq \frac{\e^{- r \, b^j}}{r \, \log(b) \, b^j}
	\]
	and the lemma follows.
\end{proof}

\noindent Recall that we are considering homogeneous activities and may write
\[
	R(B) = R_j = \prod_{k \geq j} (1 + \widehat z_k) - 1
\]
for $B \in \mathbb B_j$, $j \in \mathbb Z$. Note also that we may write
\[
	\widehat z_j = z_j \, \e^{- 2^{dj} \, p_{j - 1}} \quad \text{with} \quad p_j = \sum_{k \leq j} 2^{- dk} \log(1 + \widehat z_k)
\]
for all $j \in \mathbb Z$, cf.\ \cite[Theorem~3.1]{jansen2020hierarchical}.

\begin{lemma} \label{lem:R-asymptotics-general}
	Suppose that $z$ admits a Gibbs measure $\mathbb P$. Then $\theta^* \leq p < \infty$, $\lim_{j \to \infty} 2^{dj} \, (p - p_{j - 1}) = 0$ and
	\[
		\limsup_{j \to \infty} 2^{- dj} \, \log(R_j) = \limsup_{j \to \infty} 2^{- dj} \, \log(\widehat z_j) = \theta^* - p \in [- \infty, 0].
	\]
\end{lemma}

\begin{proof}
	Theorem~\ref{thm:existence-hom} readily implies
	\[
		0 \leq p = \sum_{j \in \mathbb Z} 2^{- dj} \, \log(1 + \widehat z_j) \leq \sum_{j \in \mathbb N} 2^{dj} \, z_{-j} + \sum_{j \in \mathbb N_0} \widehat z_j < \infty
	\]
	as well as
	\[
		0 \leq 2^{dj} \, (p - p_{j - 1}) = \sum_{k \geq j} 2^{d (j - k)} \, \log(1 + \widehat z_k) \leq \sum_{k \geq j} \widehat z_k \searrow 0
	\]
	as $j \uparrow \infty$. Furthermore,
	\begin{align*}
		- \infty \leq \theta^* & = \limsup_{j \to \infty} 2^{- dj} \, \log(z_j) \\
		& \leq \limsup_{j \to \infty} 2^{- dj} \, \log \left( \Xi_{[0, 2^j)^d}(z) \right) \\
		& = p,
	\end{align*}
	see also Lemma~\ref{lem:finite-pf}. Thus, we have
	\[
		\limsup_{j \to \infty} 2^{- dj} \, \log(\widehat z_j) = \limsup_{j \to \infty} 2^{- dj} \, \log(z_j) - p = \theta^* - p \in [- \infty, 0].
	\]
	Finally, observe that $\lim_{j \to \infty} R_j = 0$ clearly entails
	\[
		\limsup_{j \to \infty} 2^{- dj} \, \log(R_j) \leq 0
	\]
	and that the latter superior limit is equal to $\limsup_{j \to \infty} 2^{- dj} \, \log(\widehat z_j)$ by applying Lemma~\ref{lem:gamma-series-summand-equivalence} with $b = 2^d$ and the bounds
	\[
		\widehat z_j \leq \sum_{k \geq j} \widehat z_k \leq R_j = \prod_{k \geq j} (1 + \widehat z_k) - 1 \leq (1 + R_j) \, \sum_{k \geq j} \widehat z_k,
	\]
	valid for all $j \in \mathbb Z$.
\end{proof}

\noindent This concludes the proof of our first decorrelation theorem.

\begin{proof}[Proof of Theorem~\ref{thm:decay-hom-general}]
	Simply combine Lemmas~\ref{lem:covariance-explicit} and \ref{lem:R-asymptotics-general}.
\end{proof}

\noindent The proof of our second decorrelation theorem is essentially an extension of the final arguments above, using the specific choice of $z$.

\begin{proof}[Proof of Theorem~\ref{thm:decay-hom-parametric}]
	For what remains to be shown, we fix $\mu \in \mathbb R$, $J > 0$ and $\alpha \in (0, 1)$ such that $\sum_{j \in \mathbb N_0} \widehat z_j < \infty$ with
	\[
		\widehat z_j = \widehat z_j(\mu, J, \alpha) = \exp \left( 2^{dj} \, (\mu - p_{j - 1}) - 2^{\alpha dj} \, J \right) \, \1_{\{j \geq 0\}}
	\]
	for all $j \in \mathbb Z$. Recall that $\theta^* = \mu$ in the present setting. Adapting the arguments from the end of the proof of Lemma~\ref{lem:R-asymptotics-general}, we have
	\[
		\widehat z_j \, \e^{2^{dj} \, (p - \theta^*)} \leq R_j \, \e^{2^{dj} \, (p - \theta^*)} \leq (1 + R_j) \, \sum_{k \geq j} \widehat z_k \, \e^{2^{dk} \, (p - \theta^*)},
	\]
	where we used $p - \theta^* \geq 0$ in the second bound. Note that, for all $j \in \mathbb Z$,
	\[
		\widehat z_j \, \e^{2^{dj} \, (p - \theta^*)} = \exp \left( 2^{dj} \, (p - p_{j - 1}) - 2^{\alpha dj} \, J \right)
	\]
	and that, again, Lemmas~\ref{lem:covariance-explicit} and \ref{lem:R-asymptotics-general} respectively yield
	\[
		\lim_{j \to \infty} R_j = 0 \quad \text{and} \quad \lim_{j \to \infty} 2^{dj} \, (p - p_{j - 1}) = 0.
	\]
	In conclusion, we obtain
	\[
		\lim_{j \to \infty} \left( \log(R_j) + 2^{dj} \, (p - \theta^*) + 2^{\alpha dj} \, J \right) = \lim_{j \to \infty} 2^{dj} \, (p - p_{j - 1}) = 0
	\]
	from another application of Lemma~\ref{lem:gamma-series-summand-equivalence} with $r = J$ and $b = 2^{\alpha d}$.
\end{proof}

\noindent At the very end, let us briefly indicate how to extend the exponential decay from covariances of blocks to covariances of (finite sub-)configurations. Some readers might expect an exponential mixing statement resembling \cite[Corollary~7.9, Theorem~7.12, Definition~9.5]{georgiihaeggstroemmaes} but we stick with a formulation that is more in line with the theory of Gibbsian point processes.

Most of Lemma~\ref{lem:covariance-explicit} is easily generalised to the following statement: For all $\P \in \mathcal G(z)$ and all disjoint finite subsets $\mathcal B, \mathcal B' \subset \mathbb B$,
\[
	\P(\omega \supset \mathcal B \cup \mathcal B') = \1_\Delta(\mathcal B \cup \mathcal B') \, \P(\omega \supset \mathcal B) \, \P(\omega \supset \mathcal B') \, (1 + R(\mathcal B''))
\]
where $\mathcal B'' = \min(\mathbb A_\mathcal B^* \cap \mathbb A_{\mathcal B'}^*)$ and $R(\mathcal B'') = (1 + \widehat z)^{\mathbb A_{\mathcal B''}} - 1$ with
\[
	\sum_{B \in \mathbb A_{\mathcal B''}} \widehat z(B) \leq R(\mathcal B'') \leq (1 + R(\mathcal B'')) \, \sum_{B \in \mathbb A_{\mathcal B''}} \widehat z(B),
\]
\[
	\sum_{B \in \mathbb A_{\mathcal B''}} \widehat z(B) \leq |\mathcal B''| \max_{B \in \mathcal B''} \sum_{B' \in \mathbb A_B} \widehat z(B'), \quad 1 \leq |\mathcal B''| \leq \min\{|\mathcal B|, \, |\mathcal B'|\},
\]
and
\[
	1 \leq (1 + R(\mathcal B'')) \leq \prod_{B \in \mathcal B''} (1 + R(B)) \leq \max_{B \in \mathcal B''} (1 + R(B))^{|\mathcal B''|}.
\]
Assuming homogeneous $z$ and $\mathcal B \cup \mathcal B' \in \Delta$, the analysis yielding the proof of Theorem~\ref{thm:decay-hom-general} then yields
\begin{align*}
	& |\P(\omega \supset \mathcal B \cup \mathcal B') - \P(\omega \supset \mathcal B) \, \P(\omega \supset \mathcal B')| \leq \\
	& \qquad \P(\omega \supset \mathcal B) \, \P(\omega \supset \mathcal B') \, \min\{|\mathcal B|, \, |\mathcal B'|\} \, c_{\mathcal B, \mathcal B'}(\theta) \, \e^{- (p - \theta) \, D(\mathcal B, \mathcal B')}
\end{align*}
for all finite $\theta \in [\theta^*, p)$ where, writing $\mathrm{lcs}(\mathcal B, \mathcal B') = \min_{(B, B') \in \mathcal B \times \mathcal B'} \mathrm{lcs}(B, B')$,
\[
	D(\mathcal B, \mathcal B') = 2^{d \, \mathrm{lcs}(\mathcal B, \mathcal B')} \, \1_{\{\mathcal B \cap \mathcal B' = \varnothing\}} = \min_{(B, B') \in \mathcal B \times \mathcal B'} D(B, B')
\]
is just the hierarchical distance between the sets $\mathcal B$ and $\mathcal B'$ and
\[
	c_{\mathcal B, \mathcal B'}(\theta) \leq (1 + R_{\mathrm{lcs}(\mathcal B, \mathcal B')})^{\min\{|\mathcal B|, \, |\mathcal B'|\} - 1} \, c_{\mathrm{lcs}(\mathcal B, \mathcal B')}(\theta)
\]
with $\limsup_{j \to \infty} 2^{- dj} \, \log(c_j(\theta)) = \theta^* - \theta \in [- \infty, 0]$.

The above decorrelation result should be seen as an analogue to \cite[Theorem~2]{benesetal2020decorrelation}. The rigorous link consists in the fact that, in our setting of simple point processes on the discrete space $\mathbb B$, implicitly equipped with the counting measure, the sequence of point process theoretic factorial moment measures $(\alpha_p)_{p \in \mathbb N_0}$ and their densities $(\rho_p)_{p \in \mathbb N_0}$, customarily called correlation functions, can be identified with
\[
	\sum_{\substack{ \mathcal B \subset \mathbb B: \\ |\mathcal B| < \infty }} \P(\omega \supset \mathcal B) \, \delta_\mathcal B \quad \text{and} \quad \mathcal B \mapsto \P(\omega \supset \mathcal B),
\]
respectively. In our model, the latter objects are supported on finite configurations in $\Delta$.

\subsubsection*{Statement on data availability and no conflict of interest} 
Data sharing is not applicable. We do not analyse or generate any datasets, because our work proceeds within a theoretical and mathematical approach. 

The authors have no competing interests to declare that are relevant to the content of this article

\subsubsection*{Acknowledgement.} This research has been funded by the Deutsche Forschungsgemeinschaft (DFG) by grant SPP 2265 ``Random Geometric Systems'', Project P13.

\end{document}